\documentclass{amsart}
\usepackage{fullpage,graphicx,subfigure,mathpazo,color}
\usepackage{amsmath,amscd,tikz,mathrsfs,amssymb,amsfonts}
\usepackage[normalem]{ulem}
\usepackage{amsmath}
\usepackage{epstopdf}
\usepackage{tikz}
\usepackage{setspace}%Ê¹ÓÃŒäŸàºê°ü
\newcommand{\ii}{\mathrm{i}}
\newcommand{\ee}{\mathrm{e}}

\newtheorem{rhp}{Riemann-Hilbert Problem}
\newtheorem{theorem}{Theorem}
\newtheorem{lemma}{Lemma}
\newtheorem{prop}{Proposition}

\newtheorem{definition}{Definition}

\usepackage{graphicx}      % insert graphic
\usepackage{titlesec}

\titleformat{\section}{\centering\LARGE\bfseries}{\thesection}{1em}{}
\titleformat{\subsection}{\Large\bfseries}{\thesubsection}{1em}{}
\begin{document}

\title{A modified Korteweg-de Vries equation soliton gas under the nonzero background}
\author{Xiaoen Zhang}
\address{College of Mathematics and Systems Science, Shandong University of Science and Technology, Qingdao, China, 266590}
\email{zhangxiaoen@sdust.edu.cn}
\author{Liming Ling}
\address{School of Mathematics, South China University of Technology, Guangzhou, China, 510641}
\email{linglm@scut.edu.cn}
\begin{abstract}
In this paper, we consider a soliton gas of the focusing modified Korteweg-de Vries generated from the $N$-soliton solutions under the nonzero background. The spectral soliton density is chosen on the pure imaginary axis, excluding the branch cut $\Sigma_{c}=\left[-\ii, \ii\right]$. 
In the limit $N\to\infty$, we establish the Riemann-Hilbert problem of the soliton gas. Using the Deift-Zhou nonlinear steepest-descent method, this soliton gas under the nonzero background will decay to a constant background as $x\to+\infty$, while its asymptotics as $x\to-\infty$ can be expressed with a Riemann-Theta function, attached to a Riemann surface with genus-two. We also analyze the large $t$ asymptotics over the entire spatial domain, which is divided into three distinct asymptotic regions depending on the ratio $\xi=\frac{x}{t}$. Using the similar method, we provide the leading-order asymptotic behaviors for these three regions and exhibit the dynamics of large $t$ asymptotics. 


\end{abstract}
\date{\today}
\maketitle
\section{Introduction}
In this paper, we study the soliton gas for the following modified Korteweg-de Vries (mKdV) equation:
\begin{equation}\label{eq:mkdv}
q_{t}+q_{xxx}+6q^2q_{x}=0.
	\end{equation}
Due to the Lax integrability, the mKdV equation \eqref{eq:mkdv} possesses the following Lax pair,
\begin{equation}\label{eq:Lax-pair}
	\begin{aligned}
	\pmb{\Phi}_x&=\mathbf{U}(\lambda; x, t)\pmb{\Phi}, \quad \mathbf{U}(\lambda; x, t):=\begin{bmatrix}-\ii\lambda&q(x,t)\\-q(x,t)&\ii
	\lambda
	\end{bmatrix},\\
\pmb{\Phi}_{t}&=\mathbf{V}(\lambda; x, t)\pmb{\Phi},\quad \mathbf{V}(\lambda; x, t):=\begin{bmatrix}-4\ii\lambda^3+2\ii\lambda q^2&4\lambda^2q+2\ii\lambda q_{x}-q_{xx}-2q^3\\-4\lambda^2q+2\ii\lambda q_x+q_{xx}+2q^3&4\ii\lambda^3-2\ii\lambda q^2\end{bmatrix}.
\end{aligned}
	\end{equation}
The Lax pair formulation provides a comprehensive theoretical framework for studying certain initial value problems through the inverse scattering transformation, also known as the nonlinear Fourier transformation. For the mKdV equation, its inverse scattering transformation under the zero background condition was provided in the early 1973\cite{ablowitz1973nonlinear}. However, under the nonzero background, studying the mKdV equation becomes much more difficult due to the presence of a continuous spectrum branch cut and the existence of modulational instability. Under the zero or nonzero background, this equation supports many kinds of solutions, such as the soliton solutions\cite{wadati1973modified}  and the breathers\cite{zhang2020focusing}, where the soliton solutions are localized travelling waves that do not deform during the propagation. In the terminology of the inverse scattering transform, the $N$-soliton solutions of integrable equations correspond the scattering data $a(\lambda)$ having $N$ distinct simple zeros. N-soliton solutions exhibit variable characteristics, they can form ordered macroscopic coherent structures, 
%such as the dispersive shock waves. 
and they can also form irregular statistical ensembles. In \cite{lyng2007n}, the authors studied the N-soliton solutions of the nonlinear Schr\"odinger  equation for large $N$, which is converted to the semiclassical limit problem with a scale transformation. The semiclassical limit problem of some integrable systems can be regarded as the coherent nonlinear superposition of many solitons, while the incoherent or random nonlinear superpositions of solitons is related to the irregular statistical ensembles. The latter case is usually referred to as the soliton gas, a concept first proposed by Zakharov as early as 1971\cite{zakharov1971kinetic}. In that paper, a kinetic equation describing the ``rarefied" gas for the KdV equation was presented. This kinetic equation has been further extended to the case of a dense gas by considering the thermodynamic limit of the Whitham equations\cite{el2003thermodynamic}. Recently, there has been a new trend in studying the soliton gas or breather gas, including investigations into the kinetic equations\cite{congy2021soliton,el2005kinetic,tovbis2023soliton} and spectral theory\cite{el2020spectral,suret2020nonlinear}, which provides a theoretical basis for understanding the soliton gas and the long time evolution of spontaneous modulational instability. Meanwhile, there have also been new researches on the large-order asymptotics under both zero and nonzero backgrounds. This large-order solutions correspond the scattering data $a(\lambda)$ having multi-pole zeros, exhibiting different behaviors from the multi-soliton solutions. 

Towards to discovering new interesting results to the soliton gas, in \cite{dyachenko2016primitive}, Dyachenko et.al. utilized the ``the dressing method" to produce the ``primitive potentials", which can be regarded as a form of soliton gas. They also derived a Riemann-Hilbert problem for this soliton gas. Subsequently, in \cite{girotti2021rigorous}, Girotti et.al. gave another Riemann-Hilbert problem for the soliton gas and analyzed rigourous asymptotics from two aspects, one focusing on the large $x$ asymptotics at $t=0$ and the other one on the large $t$ asymptotics. For the nonlinear Schr\"odinger equation, Bertola et.al. studied the soliton shielding\cite{bertola2023soliton}, where the point spectrum is interpolated in a given spectral soliton density within a bounded domain. And to the mKdV equation, Girotti et.al. discussed the interaction between the soliton gas and a single trial soliton, deriving a kinetic velocity for the trial soliton\cite{girotti2023soliton}. It is worth noting that the aforementioned asymptotics of soliton gas are all derived from the soliton solutions under the zero background. It is well known that the integrable equations have a variety of solutions under the nonzero background, such as breathers and solitons. By comparing to the zero background, the solutions under the nonzero background behaves differently, which can be seen from Refs. \cite{wadati1973modified,zhang2020focusing}. They might interact with the background in ways that affect their speed, shape, or stability. Sometimes, these solitons may undergo phase shifts or change shape due to interactions with the background wave. In the subsequent paragraph, we provide further examples to illustrate this distinction.

From the finite-gap solutions \cite{bertola2016rogue,chen2018rogue,chen2019rogue} of Eq.\eqref{eq:mkdv}, we can discern the distinction between soliton solutions under zero and nonzero background. The soliton solutions under zero background can be reduced from the finite-gap solutions when each gap collapses into a point, while the soliton solutions under nonzero background can be reduced from the finite-gap solutions when only one gap remains finite and the each one of the other gaps collapses into a point. 
Based on this theory, we know that by comparing to the soliton solutions under the zero background, the genus of the soliton solutions under the nonzero background increases by one. This result has been corroborated by another viewpoint. For instance, in our previous study of the large-order asymptotics of Kuznetsov-Ma breather solutions, the asymptotic expression in a specific region is given by the Riemann-Theta function relating the algebraic curve of genus-two, whereas the large-order solitons is the genus-one. Considering the difference between the solutions under the zero and nonzero background, it becomes much more interesting to study the soliton gas under the nonzero background. 
%One typical characteristic of the soliton gas is the kinetic equation, which can describe the evolution of the spectral distribution function of solitons as they collide.  the kinetic equation of the soliton gas under the nonzero background is worth of discuss. 

In this paper, we are focusing on the asymptotics of the soliton gas, which is originating from the limit as $N\to\infty$ of the $N$-soliton solutions under the nonzero background. For convenience, we impose this equation with the following constant boundary condition: 
\begin{equation}\label{eq:bc}
	\lim \limits_{x\to\pm\infty} q(x,0)= 1.
	\end{equation}
For the mKdV equation \eqref{eq:mkdv}, the constant background of mKdV equation can not be reduced to the zero background case. But the constant background for the mKdV equation is equivalent to the zero background problem of the Gardner equation\cite{miura1968korteweg}. If we introduce a new variable $p=q-1$ and substitute it into Eq.\eqref{eq:mkdv}, the variable $p$ will satisfy the following equation,
\begin{equation}\label{eq:mkdv-1}
p_{t}+6p_{x}+p_{xxx}+6(2pp_{x}+p^2p_{x})=0,
\end{equation}
Eq.\eqref{eq:mkdv-1} can be further reduced to mixed KdV-mKdV equation, also known as the Gardner equation,
\begin{equation}
p_{\tilde{t}}+p_{\tilde{x}\tilde{x}\tilde{x}}+6(2pp_{\tilde{x}}+p^2p_{\tilde{x}})=0,\quad \tilde{t}=t, \quad \tilde{x}=x+6t. 
\end{equation}
We want to study the soliton gas of Eq.\eqref{eq:mkdv} under the constant background $q(x,t)=1$. Based on the above analysis, this problem is actually equivalent to studying the soliton gas of the Gardner equation Eq.\eqref{eq:mkdv-1} under the zero background, which has not yet been explored. 

This paper utilizes the Riemann-Hilbert and the Deift-Zhou nonlinear steepest-descent method to study the asymptotics. 
Before proceeding, we give a review of the inverse scattering transformation under the nonzero background. 

With the nonzero background given in Eq.\eqref{eq:bc}, the fundamental solution of the Lax pair \eqref{eq:Lax-pair} is expressed as, 
\begin{equation}
\pmb{\phi}_{\rm bg}(\lambda; x, t):=n(\lambda)\begin{bmatrix}1&\ii(\lambda-\rho(\lambda))\\\ii(\lambda-\rho(\lambda))&1
	\end{bmatrix}\ee^{-\ii\rho(\lambda)(x+4\lambda^2 t-2t)\sigma_3}:=\mathbf{E}(\lambda)\ee^{-\ii\theta(\lambda;  x, t)\sigma_3},
	\end{equation}
where $\theta(\lambda; x, t):=\rho(\lambda)(x+4\lambda^2 t-2t)$ and the eigenvalues $\pm \ii \rho(\lambda)$ are defined by the relation
$\rho(\lambda)^2=\lambda^2+1$, $\sigma_3$ is the one of the Pauli matrices, defined as
\begin{equation}\label{eq:pauli}
	\sigma_1=\begin{bmatrix}0&1\\1&0
		\end{bmatrix}, \quad \sigma_2=\begin{bmatrix}0&-\ii\\\ii&0
		\end{bmatrix},\quad \sigma_3=\begin{bmatrix}1&0\\0&-1
		\end{bmatrix}.
\end{equation}
Specifically, introduce a vertical branch cut 
$\Sigma_{c}=[-\ii, \ii]$, then $\rho(\lambda)$ is analytic in the complex $\lambda$ plane except for $\Sigma_c$ and satisfies the asymptotics $\rho(\lambda)\to\lambda+\mathcal{O}(\lambda^{-1})$ as $\lambda\to\infty$. The extra factor $n(\lambda)$ is defined as:
\begin{equation}
n(\lambda)^2=\frac{\lambda+\rho(\lambda)}{2\rho(\lambda)},\quad \lim\limits_{\lambda\to\infty}n(\lambda)=1,
\end{equation}
which ensures that the determinant of $\pmb{\phi}_{\rm bg}(\lambda; x, t)$ equals to $1$. 

For this problem, the continuous spectrum $\Gamma$ is composed of those values of $\lambda$ for which $\rho(\lambda)$ is real, that is $\Gamma=\mathbb{R}\cup \Sigma_c$. For $\lambda\in\Gamma$, the simultaneous Jost matrix solutions of the Lax pair Eq.\eqref{eq:Lax-pair} can be uniquely defined from the boundary condition
\begin{equation}
	\mathbf{J}^{\pm}(\lambda; x, t)\ee^{\ii\theta(\lambda; x, t)\sigma_3}=\mathbf{E}(\lambda)+o(1),\quad x\to\pm\infty.
	\end{equation} 
It is straightforward to see that the Jost solutions are bounded when $\lambda\in\Gamma$. For convenience, we introduce the normalized matrices $\mathbf{K}^{\pm}(\lambda; x, t):=\mathbf{J}^{\pm}(\lambda; x, t)\ee^{\ii\theta(\lambda; x, t)\sigma_3}$, which can be uniquely determined by the Volterra integral equations:
\begin{equation}
	\mathbf{K}^{\pm}(\lambda; x, t)=\mathbf{E}(\lambda)+\int_{\pm\infty}^{x}\mathbf{E}(\lambda)\ee^{-\ii\rho(\lambda)(x-y)\sigma_3}\mathbf{E}(\lambda)^{-1}\Delta \mathbf{Q}(y, t)\mathbf{K}^{\pm}(\lambda; y, t)\ee^{\ii\rho(\lambda)(x-y)\sigma_3}dy,\quad \lambda\in\Gamma,
	\end{equation} 
where $$\Delta\mathbf{Q}(x,t):=\begin{bmatrix}0&q-1\\-q+1&0
	\end{bmatrix}.$$
From this Volterra integral equation, we know that the first column $\mathbf{j}^{-,1}(\lambda; x, t)$ of the Jost solution $\mathbf{J}^{-}(\lambda; x, t)$ and the second column $\mathbf{j}^{+, 2}(\lambda; x, t)$ of $\mathbf{J}^{+}(\lambda; x, t)$ are analytic in the domain $\mathbb{C}^{+}\setminus \Sigma_c$. Conversely, the second column $\mathbf{j}^{-,2}(\lambda; x, t)$ of $\mathbf{J}^{-}(\lambda; x, t)$ and the first column $\mathbf{j}^{+,1}(\lambda; x, t)$ of $\mathbf{J}^{+}(\lambda; x, t)$ are analytic in the domain $\mathbb{C}^{-}\setminus \Sigma_c$.  

Since the matrices $\mathbf{U}(\lambda; x, t), \mathbf{V}(\lambda; x, t)$ are traceless, by the Abel theorem, the determinant of the Jost solutions $\mathbf{J}^{\pm}(\lambda; x, t)$ are constant. From the asymptotic boundary condition at $x\to\pm\infty$, we have $\det(\mathbf{J}^{\pm}(\lambda; x, t))=1$. Furthermore, $\mathbf{J}^{\pm}(\lambda; x, t)$ are both the fundamental solutions of the Lax pair for $\lambda\in\Gamma$, there exists a scattering matrix $\mathbf{S}(\lambda)$ such that $\mathbf{J}^{+}(\lambda; x, t)=\mathbf{J}^{-}(\lambda; x, t)\mathbf{S}(\lambda)$. Let us suppose that the scattering matrix $\mathbf{S}(\lambda)$ has the following formula:
\begin{equation}
	\mathbf{S}(\lambda):=\begin{bmatrix}\bar{a}(\lambda)&\bar{b}(\lambda)\\
		-b(\lambda)&a(\lambda)
		\end{bmatrix}.
	\end{equation} 
Using the scattering relation, the scattering data can be given as:
\begin{equation}
	\begin{aligned}
	a(\lambda):&=\det\left(\left[\mathbf{j}^{-,1}(\lambda; x, t); \mathbf{j}^{+,2}(\lambda; x, t)\right]\right),\\
	\bar{a}(\lambda):&=\det\left(\left[\mathbf{j}^{+,1}(\lambda; x, t); \mathbf{j}^{-,2}(\lambda; x, t)\right]\right),\\
	b(\lambda):&=\det\left(\left[\mathbf{j}^{+,1}(\lambda; x, t); \mathbf{j}^{-,1}(\lambda; x, t)\right]\right),\\
	\bar{b}(\lambda):&=\det\left(\left[\mathbf{j}^{+,2}(\lambda; x, t); \mathbf{j}^{-,2}(\lambda; x, t)\right]\right).
	\end{aligned}
	\end{equation}
From the analyticity of the Jost solution,  we know that the scattering data $a(\lambda), \bar{a}(\lambda)$ can be analytically extended to the domains $\lambda\in\mathbb{C}^+\setminus \Sigma_c$ and $\lambda\in\mathbb{C}^-\setminus \Sigma_c$ respectively. This property is crucial for the later analysis of soliton gas under the nonzero background. 

Assuming for convenience that $a(\lambda)\neq 0$ for all $\lambda\in\Gamma\setminus \{-\ii, \ii\}, $ and in the domain $\mathbb{C}^{+}\setminus \Sigma_c$, $a(\lambda)$ has simple zeros at $\lambda=\xi_{1}, \xi_2, \cdots, \xi_N$. In this paper, we further suppose that these zeros satisfy the constraint condition $\xi_{i}(i=1,2,\cdots N)\in \ii\mathbb{R}^+\setminus\Sigma_{c}$. At these points, there exist the nonzero proportionality constants $\gamma_1, \gamma_2, \cdots, \gamma_N$, such that 
\begin{equation}\label{eq:propor-constant}
	\mathbf{j}^{-,1}(\xi_{j}; x, t)=\gamma_j\mathbf{j}^{+,2}(\xi_{j}; x, t), \quad (j=1,2,\cdots, N). 
	\end{equation}
Next, we give the symmetry of the Jost solution and the scattering data. 

\emph{Symmetry relation} Since $q(x, t)$ is a real function, it is  straightforward to see that the coefficient matrices $\mathbf{U}(\lambda; x, t)$ and $\mathbf{V}(\lambda; x, t)$ in the Lax pair Eq.\eqref{eq:Lax-pair} satisfy the symmetries:
\begin{equation}
\begin{aligned}
	\mathbf{U}(-\lambda; x, t)=\mathbf{U}(\lambda^*; x, t)^*,\quad \mathbf{U}(\lambda; x, t)=\sigma_2\mathbf{U}(\lambda^*; x, t)^*\sigma_2,\\
		\mathbf{V}(-\lambda; x, t)=\mathbf{V}(\lambda^*; x, t)^*,\quad \mathbf{V}(\lambda; x, t)=\sigma_2\mathbf{V}(\lambda^*; x, t)^*\sigma_2.
\end{aligned}
	\end{equation}
Then the Jost solution will also satisfy these symmetries $$\mathbf{J}(-\lambda; x, t)=\mathbf{J}(\lambda^*; x, t)^*,\quad \mathbf{J}(\lambda; x, t)=\sigma_2\mathbf{J}(\lambda^*; x, t)^*\sigma_2, \quad \lambda\in\Gamma\setminus \left\{-\ii, \ii\right\}.$$
These symmetries also extend to the scattering data, implying $a(-\lambda)=a(\lambda^*)^*, \bar{a}(-\lambda)=\bar{a}(\lambda^*)^*, b(-\lambda)=b(\lambda^*)^*, \bar{b}(-\lambda)=\bar{b}(\lambda^*)^*, \bar{a}(\lambda; t)=a(\lambda^*; t)^*, \bar{b}(\lambda; t)=b(\lambda^*; t)^*$ at $\lambda\in\Gamma\setminus\left\{-\ii, \ii\right\}.$ Moreover, the symmetry $\mathbf{J}(-\lambda; x, t)=\mathbf{J}(\lambda^*; x, t)^*$ indicates $\mathbf{j}^{-,1}(-\lambda; x, t)=\mathbf{j}^{-,1}(\lambda^*; x, t)^*$ and $\mathbf{j}^{+,2}(-\lambda; x, t)=\mathbf{j}^{+,2}(\lambda^*; x, t)^*$. Substituting this symmetry to the relation \eqref{eq:propor-constant}, the proportionality constants $\gamma_j(j=1,2,\cdots, N)$ satisfy the symmetry 
\begin{equation}\gamma_j=\frac{\mathbf{j}^{-,1}(\xi_j; x, t)}{\mathbf{j}^{+,2}(\xi_j; x, t)}=\frac{\mathbf{j}^{-,1}(-\xi_j^*; x,t)^*}{\mathbf{j}^{+,2}(-\xi_j^*; x, t)^*}. \end{equation}
At the simple pure imaginary zeros, that is $\xi_j=-\xi_j^*(j=1,2,\cdots, N)$, we have the relation $\gamma_j=\gamma_j^*$, thus the proportionality constants $\xi_j$ are all real. 

\emph{Inverse scattering transform} The Beals-Coifman simultaneous solution to the Lax pair Eq.\eqref{eq:Lax-pair} is the following sectional meromorphic function,
\begin{equation}
	\pmb{\phi}^{\rm BC}(\lambda; x, t):=\left\{\begin{aligned}&\left[a(\lambda)^{-1}\mathbf{j}^{-,1}(\lambda; x, t);\quad \mathbf{j}^{+,2}(\lambda; x, t)\right],\quad \lambda\in\mathbb{C}^+\setminus \Sigma_c,\\
		&\left[\mathbf{j}^{+,1}(\lambda; x, t);\quad \bar{a}(\lambda)^{-1}\mathbf{j}^{-,2}(\lambda; x, t)\right],\quad\lambda\in\mathbb{C}^-\setminus \Sigma_c.
		\end{aligned}\right.
	\end{equation}
Then the related normalized matrix value function is set as:
\begin{equation}
	\mathbf{M}^{\rm BC}(\lambda; x, t):=\pmb{\phi}^{BC}(\lambda; x, t)\ee^{\ii\theta(\lambda; x, t)\sigma_3},\quad \lambda\in \Gamma,
	\end{equation}
obviously, $\mathbf{M}^{\rm BC}(\lambda; x, t)$ has simple poles at the points $\left\{\xi_1, \xi_2, \cdots, \xi_N\right\}$ as well as their conjugates $\left\{\xi_1^*, \xi_2^*, \cdots, \xi_N^*\right\}.$ It also satisfies the normalization condition at $\lambda=\infty$,
\begin{equation}
	\lim\limits_{\lambda\to\infty}\mathbf{M}^{\rm BC}(\lambda; x, t)=\mathbb{I}.
	\end{equation}
\emph{Jump conditions of $\mathbf{M}^{\rm BC}(\lambda; x, t)$ across the continuous spectrum $\Gamma$} For $\lambda\in\mathbb{R}$, the continuous boundary values of $\mathbf{M}^{\rm BC}(\lambda; x, t)$ are defined as:
\begin{equation}
	\mathbf{M}^{\rm BC}_{\pm}(\lambda; x, t):=\lim\limits_{\epsilon\to 0}\mathbf{M}^{\rm BC}(\lambda\pm\ii\epsilon; x, t).
	\end{equation}
Using the scattering condition $\mathbf{J}^{+}(\lambda; x, t)=\mathbf{J}^{-}(\lambda; x, t)\mathbf{S}(\lambda)$, we obtain the jump condition on the real axis,
\begin{equation}
	\mathbf{M}^{\rm BC}_{+}(\lambda; x, t)=\mathbf{M}^{\rm BC}_{-}(\lambda; x, t)\ee^{-\ii\theta(\lambda; x, t)\sigma_3}\mathbf{V}^{\mathbb{R}}(\lambda)\ee^{\ii\theta(\lambda; x, t)\sigma_3},
	\end{equation}
where $$\mathbf{V}^{\mathbb{R}}(\lambda):=\begin{bmatrix}1+|R(\lambda)|^2&R(\lambda)^*\\R(\lambda)&1
	\end{bmatrix}, \quad R(\lambda):=\frac{b(\lambda)}{a(\lambda)}.$$
For $\lambda\in\Sigma_c$, the boundary values of $\mathbf{M}^{\rm BC}(\lambda; x, t)$ are defined as
\begin{equation}
	\mathbf{M}^{\rm BC}_{\pm}(\lambda; x, t):=\lim\limits_{\epsilon\to 0}\mathbf{M}^{\rm BC}(\lambda\pm\epsilon; x, t),
	\end{equation}
then we have the following jump condition for $\lambda\in\Sigma_c$, 
\begin{equation}
	\mathbf{M}^{\rm BC}_{+}(\lambda; x, t)=\mathbf{M}^{\rm BC}_{-}(\lambda; x, t)\ee^{-\ii\rho_{-}(\lambda)(x+4\lambda^2 t-2t)\sigma_3}\mathbf{V}^{\downarrow}(\lambda)\ee^{\ii\rho_{+}(\lambda)(x+4\lambda^2 t-2t)\sigma_3},\quad \lambda\in\ii\mathbb{R}^{+}, 0<\Im(\lambda)<1,
	\end{equation}
where 
\begin{equation*}
	\mathbf{V}^{\downarrow}(\lambda):=\begin{bmatrix}\ii \bar{R}_-(\lambda)&-\ii\\-\ii\left(1+R_{-}(\lambda)\bar{R}_{-}(\lambda)\right)&\ii R_{-}(\lambda),
		\end{bmatrix}
	\end{equation*}
and $R_{-}(\lambda)=\lim\limits_{\epsilon\to 0}R(\lambda-\epsilon), \bar{R}_{-}(\lambda)=\lim\limits_{\epsilon\to 0}\bar{R}(\lambda-\epsilon).$ 
On the lower half-plane, we have
\begin{equation}
	\mathbf{M}^{\rm BC}_{+}(\lambda; x, t)=\mathbf{M}^{\rm BC}_{-}(\lambda; x, t)\ee^{-\ii\rho_{-}(\lambda)(x+4\lambda^2 t-2t)\sigma_3}\mathbf{V}^{\uparrow}(\lambda)\ee^{\ii\rho_{+}(\lambda)(x+4\lambda^2 t-2t)\sigma_3},\quad \lambda\in\ii\mathbb{R}^{-},  -1<\Im(\lambda)<0,
\end{equation}
where 
\begin{equation}
	\mathbf{V}^{\uparrow}(\lambda)=\begin{bmatrix}
	-\ii\bar{R}_{-}(\lambda)&-\ii(1+R_{-}(\lambda)\bar{R}_-(\lambda))\\-\ii&-\ii R_{-}(\lambda)
	\end{bmatrix}.
	\end{equation}
At the simple poles $\xi_1, \xi_2, \cdots, \xi_N$, we obtain the following residue condition,
\begin{equation}\label{eq:res}
	\begin{aligned}
&\mathop{\rm Res}\limits_{\lambda=\xi_j} \mathbf{M}^{\rm BC}(\lambda; x, t)=\lim\limits_{\lambda\to\xi_j}\mathbf{M}^{\rm BC}(\lambda; x, t)\ee^{-\ii\theta(\xi_j; x, t)\sigma_3}
\begin{bmatrix}0&0\\
	a'(\xi_j)^{-1}\gamma_j&0
	\end{bmatrix}\ee^{\ii\theta(\xi_j; x, t)\sigma_3},\quad &&j=1,2,\cdots, N,\\
&\mathop{\rm Res}\limits_{\lambda=\xi_j^*} \mathbf{M}^{\rm BC}(\lambda; x, t)=\lim\limits_{\lambda\to\xi_j^*}\mathbf{M}^{\rm BC}(\lambda; x, t)\ee^{-\ii\theta(\xi_j^*; x, t)\sigma_3}
\begin{bmatrix}0&-\bar{a}'(\xi_j^*)^{-1}\gamma_j\\
	0&0
\end{bmatrix}\ee^{\ii\theta(\xi_j^*; x, t)\sigma_3},\quad&& j=1,2,\cdots, N.
\end{aligned}
	\end{equation}
In this paper, we aim to study the asymptotics of the soliton gas, which can be generated from the $N$-soliton solutions as $N\to\infty$. From the jump condition of $\mathbf{M}(\lambda; x, t)$ when $\lambda\in\Gamma$ and the residue condition Eq. \eqref{eq:res}, we can derive the exact $N$-soliton solutions under the nonzero background by choosing $R(\lambda)=0$. In the next section, we will give a detailed description of how to derive a Riemann-Hilbert problem of the soliton gas by taking $N\to\infty$. 
\section{Riemann-Hilbert problem of the soliton gas under the nonzero background}
\label{sec:RHP}
From the jump condition in the continuous spectrum $\Gamma$ and the residue condition Eq.\eqref{eq:res} at these simple poles, we derive the Riemann-Hilbert problem for pure soliton solutions under the nonzero background. In this case, the reflection coefficient $R(\lambda)$ vanishes, and the simple poles $\xi_{j}(j=1,2,\cdots, N)$ are chosen to be pure imaginary. Then we can find $2\times 2$ matrix $\mathbf{M}(\lambda; x, t)$ satisfying the following Riemann-Hilbert problem.
\begin{rhp}\label{rhp:residue}The Riemann-Hilbert problem for the $N$-soliton solutions under the nonzero background.
	\begin{itemize}
	\item $\mathbf{M}(\lambda; x, t)$ is a meromorphic matrix function in $\mathbb{C}\setminus\Sigma_c$, with simple poles at the points $\xi_1, \xi_2, \cdots, \xi_N$ and their conjugates $\xi_1^*, \xi_2^*, \cdots, \xi_N^*.$
	\item $\mathbf{M}(\lambda; x, t)$ satisfies the following residue condition, 
\begin{equation}
	\begin{aligned}
		&\mathop{\rm Res}\limits_{\lambda=\xi_j} \mathbf{M}(\lambda; x, t)=\lim\limits_{\lambda\to\xi_j}\mathbf{M}(\lambda; x, t)\ee^{-\ii\theta(\xi_j; x, t)\sigma_3}
		\begin{bmatrix}0&0\\
			a'(\xi_j)^{-1}\gamma_j&0
		\end{bmatrix}\ee^{\ii\theta(\xi_j; x, t)\sigma_3},\quad &&j=1,2,\cdots, N,\\
		&\mathop{\rm Res}\limits_{\lambda=\xi_j^*} \mathbf{M}(\lambda; x, t)=\lim\limits_{\lambda\to\xi_j^*}\mathbf{M}(\lambda; x, t)\ee^{-\ii\theta(\xi_j^*; x, t)\sigma_3}
		\begin{bmatrix}0&-\bar{a}'(\xi_j^*)^{-1}\gamma_j\\
			0&0
		\end{bmatrix}\ee^{\ii\theta(\xi_j^*; x, t)\sigma_3},\quad &&j=1,2,\cdots, N.
	\end{aligned}
\end{equation}	
\item For $\lambda\in\Sigma_c$, $\mathbf{M}(\lambda; x, t)$ satisfies the jump condition,
\begin{equation}
	\mathbf{M}_{+}(\lambda; x, t)=\mathbf{M}_{-}(\lambda; x, t)\begin{bmatrix}0&\ii\\\ii&0
		\end{bmatrix}, \quad\mathbf{M}_{\pm}(\lambda; x, t):=\lim\limits_{\epsilon\to 0}\mathbf{M}(\lambda\mp\epsilon; x, t).
	\end{equation}
\item As $\lambda\to\infty$, $\mathbf{M}(\lambda; x, t)$ satisfies the normalization condition,
\begin{equation}
	\mathbf{M}(\lambda; x, t)\to\mathbb{I}\quad \text{as}\quad \lambda\to\infty.
	\end{equation}
		\end{itemize}
	Then the potential $q(x, t)$ can be recovered through a limit,
	\begin{equation}
		q(x, t)=2\ii\lim\limits_{\lambda\to\infty}\lambda\mathbf{M}(\lambda; x, t)_{12}.
		\end{equation}
	\end{rhp}
When $N=1$, the solution of the RHP \ref{rhp:residue} can be given as
\begin{equation}
	\mathbf{M}(\lambda; x, t)=\left(\mathbb{I}+\frac{\mathbf{A}(x,t)}{\lambda-\lambda_1}+\frac{\mathbf{B}(x,t)}{\lambda-\lambda_1^*}\right)\mathbf{E}(\lambda).
	\end{equation}
Substituting it to RHP \ref{rhp:residue}, using the residue condition, we have the following relations:
\begin{equation}\label{eq:linear-relation}
	\begin{aligned}
	&\mathbf{A}(x,t)\mathbf{E}(\xi_1)=\lim\limits_{\lambda\to\xi_1}\mathbf{M}(\lambda; x, t)\ee^{-\ii\theta(\xi_1; x, t)\sigma_3}
	\begin{bmatrix}0&0\\
		a'(\xi_1)^{-1}\gamma_j&0
	\end{bmatrix}\ee^{\ii\theta(\xi_1; x, t)\sigma_3},\\
&\mathbf{B}(x,t)\mathbf{E}(\xi_1^*)=\lim\limits_{\lambda\to\xi_1^*}\mathbf{M}(\lambda; x, t)\ee^{-\ii\theta(\xi_1^*; x, t)\sigma_3}
\begin{bmatrix}0&-\bar{a}'(\xi_1^*)^{-1}\gamma_1\\
	0&0
\end{bmatrix}\ee^{\ii\theta(\xi_j^*; x, t)\sigma_3}.
\end{aligned}
	\end{equation}
The left side of Eq.\eqref{eq:linear-relation} is analytic at $\lambda=\xi_1, \lambda=\xi_1^*$ respectively. Taking the laurent expansion of the right-hand side hand, by comparing to the coefficients of $(\lambda-\xi_1)^{-1},  (\lambda-\xi_1)^0, (\lambda-\xi_1^*)^{-1}, (\lambda-\xi_1^*)^0$, we can get a linear equation about the coefficient matrices $\mathbf{A}(x,t), \mathbf{B}(x,t)$, which can be solved uniquely.  By choosing $\xi_1=2\ii, a'(\xi_j)\gamma_j=\beta\ii$, the one soliton under the nonzero background can be given as:
\begin{equation}
	q^{[1]}(x,t)=1+\frac{6}{1-2\cosh\left(2\sqrt{3}(x-18t)+\log(\frac{3}{\beta})\right)}.
\end{equation}
We aim to study the soliton gas under the nonzero background, defined as the limit of $N$-soliton solutions as $N\to\infty$. The Riemann-Hilbert problem \ref{rhp:residue} provides the $N$-soliton solutions under the nonzero background. Next, by incorporating additional assumptions regarding the discrete spectrum $\xi_j$ and the scattering data $a(\lambda)$, we construct a new Riemann-Hilbert problem by taking $N\to\infty$.
\begin{itemize}
	\item
	The simple poles $\left\{\xi_j\right\}_{j=1}^{N}$ are sampled with a density function $\varrho$ such that for each $j=1,2,\cdots, N$, we have the integral identity: 
	$\int_{\ii \eta_1}^{\ii \xi_j}\varrho(\eta)d\eta=\frac{j}{N}$.
	\item The coefficients $a'(\xi_j)^{-1}\gamma_j$ are assumed to be sampled into the formula, 
	\begin{equation}\label{eq:spectral}
		a'(\xi_j)^{-1}\gamma_j=\frac{\ii\left(\eta_2-\eta_1\right)r_1(\xi_j)}{N\pi},\quad j=1,2,\cdots, N,
		\end{equation}
where $r_{1}(\lambda)$ denotes an analytic non-vanishing function in the neighbourhood of these intervals $\left[\ii\eta_1, \ii\eta_2\right]$ and $\left[-\ii\eta_1, -\ii\eta_2\right]$ and satisfies the symmetry relation $r_1^*(\lambda)=r_1(\lambda^*)$. 
To study the soliton gas under the nonzero background using the Riemann-Hilbert problem, it is convenient to convert the residue condition into the jump condition. Let $\gamma_+$ be a closed curve enclosing all the simple poles $\xi_1, \xi_2,\cdots, \xi_N$ but not intersecting with the branch cut $\Sigma_c$, and $\gamma_{+}$ is clockwise in the upper half plane. $\gamma_-$ is the reflection of $\gamma_{+}$ across the real axis, maintaining a clockwise orientation in the lower half plane.  
Define a new matrix function $\mathbf{N}(\lambda; x, t)$ as:
\begin{equation}
	\mathbf{N}(\lambda; x, t):=\left\{\begin{aligned}&\mathbf{M}(\lambda; x, t)\begin{bmatrix}1&0\\
			-\sum\limits_{j=1}^{N}\frac{\ii(\eta_2-\eta_1)r_1(\xi_j)}{N\pi(\lambda-\xi_j)}\ee^{2\ii\theta(\lambda; x, t)}&1
		\end{bmatrix},\quad&&\lambda\in\gamma_{+},\\
		&\mathbf{M}(\lambda; x, t)\begin{bmatrix}1&-\sum\limits_{j=1}^{N}\frac{\ii(\eta_2-\eta_1)r_1(\xi_j^*)}{N\pi(\lambda-\xi_j^*)}\ee^{-2\ii\theta(\lambda; x, t)}\\0&1
		\end{bmatrix},\quad&& \lambda\in\gamma_{-},\\
	&\mathbf{M}(\lambda; x, t),\quad &&\text{otherwise}.
	\end{aligned}\right.
\end{equation}
\end{itemize}
With the definition of the definite integral, we can transform the above sum formula into an integral form. 
\begin{prop}
When $N\to\infty$, the sum formula is equivalent to an integral form for $\lambda\in \mathbb{C}\setminus \mathbb{K}_{+}$, where $\mathbb{K}_{+}$ is an open set that includes the interval $\left[\ii\eta_1, \ii\eta_2\right]$,
\begin{equation}		\lim\limits_{N\to\infty}\sum\limits_{j=1}^N\frac{\ii(\eta_2-\eta_1)r_1(\xi_j)}{N\pi(\lambda-\xi_j)}=\int_{\ii\eta_1}^{\ii\eta_2}\frac{r_1(\xi)}{\pi(\lambda-\xi)}d\xi.
	\end{equation}
For $\lambda\in\mathbb{C}\setminus \mathbb{K}_{-}, $ where $\mathbb{K}_{-}$ is the reflection of $\mathbb{K}_{+}$ along the real axis, we have 
\begin{equation}
	\lim\limits_{N\to\infty}\sum\limits_{j=1}^N\frac{\ii(\eta_2-\eta_1)r_1(\xi_j^*)}{N\pi(\lambda-\xi_j^*)}=\int_{-\ii\eta_2}^{-\ii\eta_1}\frac{r_1(\xi)}{\ii\pi(\lambda-\xi)}d\xi.
	\end{equation}
\end{prop}
\begin{proof}
A direct calculation yields the following relation, 
\begin{equation}
	\lim\limits_{N\to\infty}\sum\limits_{j=1}^N\frac{\ii(\eta_2-\eta_1)r_1(\xi_j)}{N\pi(\lambda-\xi_j)}=	\lim\limits_{N\to\infty}\sum\limits_{j=1}^N\frac{\ii\eta_2-\ii\eta_1}{N}\frac{r_1(\xi_j)}{\pi(\lambda-\xi_j)}=	\lim\limits_{N\to\infty}\sum\limits_{j=1}^N\Delta \xi\frac{r_1(\xi_j)}{\pi(\lambda-\xi_j)}=\int_{\ii\eta_1}^{\ii\eta_2}\frac{r_1(\xi)}{\pi(\lambda-\xi)}d\xi.
	\end{equation}
The last one follows from the definition of the definite integral, converting the Riemann sum to the Riemann Stieltjes integral. 
Similarly, for $\lambda\in\mathbb{C}\setminus \mathbb{K}_{-}$, we have 
\begin{equation}
	\lim\limits_{N\to\infty}\sum\limits_{j=1}^N\frac{\ii(\eta_2-\eta_1)r_1(\xi_j^*)}{N\pi(\lambda-\xi_j^*)}=	\lim\limits_{N\to\infty}\sum\limits_{j=1}^N\frac{-\ii\eta_1-(-\ii\eta_2)}{N}\frac{r_1(\xi_j^*)}{\pi(\lambda-\xi_j^*)}=	\lim\limits_{N\to\infty}\sum\limits_{j=1}^N\Delta\xi\frac{r_1(\xi_j^*)}{\pi(\lambda-\xi_j^*)}=\int_{-\ii\eta_2}^{-\ii\eta_1}\frac{r_1(\xi)}{\pi(\lambda-\xi)}d\xi,
	\end{equation}
	\end{proof}
Building upon the above Proposition, we establish a new Riemann-Hilbert problem for $\mathbf{N}(\lambda; x, t)$.
\begin{rhp}Seek a $2\times 2$ matrix function $\mathbf{N}(\lambda; x, t)$ that satisfies the following conditions.
	\begin{itemize}
		\item 
{\bf Analyticity:}   $\mathbf{N}(\lambda; x, t)$ is analytic for $\lambda\in\mathbb{C}\setminus \left(\Sigma_c\ \cup \partial \gamma_+\cup\partial \gamma_-\right).$
\item {\bf Jump condition:} $\mathbf{N}(\lambda; x, t)$ takes continuous boundary values $\mathbf{N}_{\pm}(\lambda; x, t)$ on $\partial \gamma_{\pm}\cup \Sigma_c$, and they are related by the following jump conditions,
\begin{equation}
	\begin{aligned}
	&\mathbf{N}_{+}(\lambda; x, t)=\mathbf{N}_{-}(\lambda; x, t)\begin{bmatrix}1&0\\\int_{\ii\eta_1}^{\ii\eta_2}\frac{r_1(\xi)}{\pi(\lambda-\xi)}d\xi\ee^{2\ii\theta(\lambda; x, t)}&1
		\end{bmatrix},\qquad\quad&& \lambda\in\partial \gamma_+,\\
	&\mathbf{N}_{+}(\lambda; x, t)=\mathbf{N}_{-}(\lambda; x, t)\begin{bmatrix}1&\int_{-\ii\eta_2}^{-\ii\eta_1}\frac{r_1(\xi)}{\pi(\lambda-\xi)}d\xi\ee^{-2\ii\theta(\lambda; x, t)}\\0&1
		\end{bmatrix},\quad&&\lambda\in\partial \gamma_-,\\
	&\mathbf{N}_{+}(\lambda;  x, t)=\mathbf{N}_{-}(\lambda; x, t)\begin{bmatrix}0&\ii\\
		\ii&0
		\end{bmatrix},\qquad\qquad\qquad\qquad\qquad\qquad\quad\quad\,&&\lambda\in\Sigma_c.
	\end{aligned}
	\end{equation}  
\item 
{\bf Normalization:} When $\lambda\to\infty$, we have
\begin{equation}
	\lim\limits_{\lambda\to\infty}\mathbf{N}(\lambda; x, t)\to\mathbb{I}.
	\end{equation}
\end{itemize}
	\end{rhp}
To eliminate the jump conditions on the boundaries of $\gamma_{\pm}$ and convert them to the intervals $\left[\ii\eta_1, \ii\eta_2\right]$ and $\left[-\ii\eta_1, -\ii\eta_2\right]$, we introduce a new matrix function $\mathbf{P}(\lambda; x, t)$ defined as:
\begin{equation}\label{eq:jump-P}
\mathbf{P}(\lambda; x, t):=	\left\{
\begin{aligned}&\mathbf{N}(\lambda; x, t)\begin{bmatrix}1&0\\\int_{\ii\eta_1}^{\ii\eta_2}\frac{r_1(\xi)}{\pi(\lambda-\xi)}d\xi\ee^{2\ii\theta(\lambda; x, t)}&1
\end{bmatrix},\qquad\quad &&\lambda\in\gamma_+,\\
&\mathbf{N}(\lambda; x, t)\begin{bmatrix}1&\int_{-\ii\eta_2}^{-\ii\eta_1}\frac{r_1(\xi)}{\pi(\lambda-\xi)}d\xi\ee^{-2\ii\theta(\lambda; x, t)}\\0&1
\end{bmatrix},\quad&&\lambda\in \gamma_-,\\
&\mathbf{N}(\lambda; x, t),\quad &&\text{otherwise}.
\end{aligned}
\right.
	\end{equation}  
It is evident that $\mathbf{P}(\lambda; x, t)$ has no jump at $\partial \gamma_{\pm}$, however, it produces a new jump at $\lambda\in\left[\ii\eta_1, \ii\eta_2\right]$ and $\lambda\in\left[-\ii\eta_1,-\ii\eta_2\right]$. By applying the Plemelj formula, the jump conditions of $\mathbf{P}(\lambda; x, t)$ are as follows, 
\begin{equation}\label{eq:jump-P-1}
\begin{aligned}
	&\mathbf{P}_{+}(\lambda; x, t)=\mathbf{P}_{-}(\lambda; x, t)\begin{bmatrix}1&0\\-2\ii r_1(\lambda)\ee^{2\ii\theta(\lambda; x, t)}&1,
		\end{bmatrix},\quad &&\lambda\in\left[\ii\eta_1, \ii\eta_2\right],\\
	&\mathbf{P}_{+}(\lambda; x, t)=\mathbf{P}_{-}(\lambda; x, t)\begin{bmatrix}
		1&-2\ii r_1(\lambda)\ee^{-2\ii\theta(\lambda; x, t)}\\
		0&1
		\end{bmatrix},\quad &&\lambda\in\left[-\ii\eta_1,-\ii\eta_2\right],\\
	&\mathbf{P}_{+}(\lambda; x, t)=\mathbf{P}_{-}(\lambda)\begin{bmatrix}0&\ii\\\ii&0
		\end{bmatrix},\quad&&\lambda\in\Sigma_c.
	\end{aligned}
	\end{equation}
Here, let us assume that all the contours $\left[\ii\eta_1, \ii\eta_2\right], \left[-\ii\eta_1, -\ii\eta_2\right], \Sigma_c$ are oriented upwards. Based on the jump conditions of $\mathbf{P}(\lambda; x, t)$, 
we find that it is well-suited for steepest-descent analysis for large $x$ and large $t$. For the large $x$ asymptotics, we should introduce a new variable $\tilde{\xi}=\frac{t}{x}$ and  rewrite the phase term $\theta(\lambda; x, t)$ as $\theta(\lambda; x, t):=x\left(\rho(\lambda)\left(1+(4\lambda^2-2)\tilde{\xi}\right)\right)$. For the large $t$ asymptotics, we introduce a different variable $\xi=\frac{x}{t}$ and rewrite the phase term $\theta(\lambda; x, t)$ as $\theta(\lambda; x, t)=t\left(\rho(\lambda)\left(\xi+4\lambda^2-2\right)\right)$. In the next two sections, we will study these two distinct asymptotic behaviors. 
 For the large $x$ asymptotics, we fix $t=0$ and study the large $x$ behavior. For the large $t$ asymptotics, we study the large $t$ behavior to distinguish it from the first one. In the first case, the large $x$ asymptotics includes two asymptotic lines; and in the second case, the large $t$ asymptotics includes three asymptotic regions, which can be vividly described by the following sketch in Fig. \ref{fig:xtplane}.
 \begin{figure}[ht]
 	\centering
 	\includegraphics[width=0.8\textwidth]{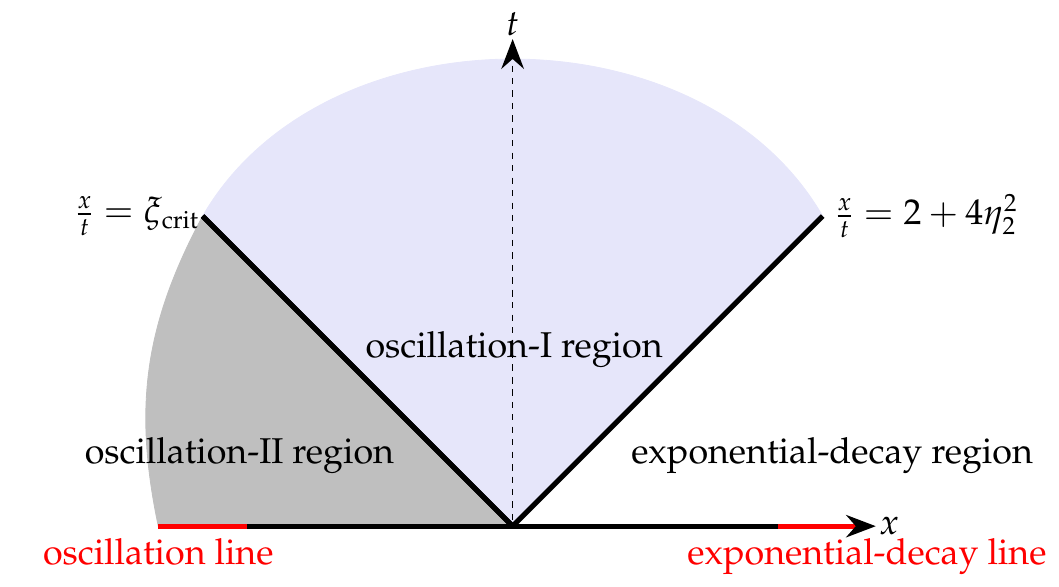}
 	\caption{Above the $x$-axis, there are two oscillation regions and one exponential-decay region in the $(x,t)$-plane, which is divided by the values of the variable $\frac{x}{t}$, where $\xi_{\rm crit}$ is defined in Eq.\eqref{eq:xi-crit} and $\eta_2$ is defined in Eq.\eqref{eq:spectral}. On the $x$-axis, there is one oscillation line and one exponential-decay line.}
 	\label{fig:xtplane}
 \end{figure}
\section{Large $x$ asymptotics of the potential $q(x,t)$ at $t=0$}
\label{sec:largex}
In this section, we prepare to study the large $x$ asymptotics at the fixed $t=0$ via the Deift-Zhou nonlinear steepest-descent method, which involves two subcases, one is the asymptotics as $x\to+\infty$ and the other one is $x\to-\infty$. When $t=0$, the jump conditions of $\mathbf{P}(\lambda; x, t)$ become simple ones, 
\begin{equation}\label{eq:jump-P}
	\begin{aligned}
		&\mathbf{P}_{+}(\lambda; x, 0)=\mathbf{P}_{-}(\lambda; x, 0)\begin{bmatrix}
			1&0\\
			-2\ii r_1(\lambda)\ee^{2\ii\rho(\lambda)x}&1
		\end{bmatrix},\quad&&\lambda\in\left[\ii\eta_1, \ii\eta_2\right],\\
		&\mathbf{P}_{+}(\lambda; x, 0)=\mathbf{P}_{-}(\lambda; x, 0)\begin{bmatrix}
			1&-2\ii r_1(\lambda)\ee^{-2\ii\rho(\lambda)x}\\
			0&1
		\end{bmatrix},\quad&&\lambda\in\left[-\ii\eta_1, -\ii\eta_2\right],\\
		&\mathbf{P}_{+}(\lambda; x, 0)=\mathbf{P}_{-}(\lambda; x, 0)\begin{bmatrix}0&\ii\\\ii&0
		\end{bmatrix},\quad&&\lambda\in\Sigma_c.
	\end{aligned}
\end{equation}
Then we give the detailed asymptotic analysis in the next two subsections. 
\subsection{The asymptotics of the potential $q(x, 0)$ as $x\to+\infty$}
In this subsection, we will study the asymptotics as $x\to+\infty$. Before studying the asymptotics, we first present a lemma to demonstrate some properties of $\mathbf{P}(\lambda; x, 0)$. 
\begin{lemma}
	When $x\to+\infty$, the first two jump matrices in Eq.\eqref{eq:jump-P} approach to the identity matrix exponentially.
\end{lemma}
\begin{proof}
	Set $\lambda=\ii b(1<\eta_1<b<\eta_2)$, with the asymptotics of $\rho(\lambda)\to\lambda$ as $\lambda\to\infty$, the factor $\ee^{2\ii\rho(\lambda)x}$ transforms into $\ee^{-2\sqrt{b^2-1}x}$, which exponentially decays as $x\to+\infty$. Additionally, due to the symmetry property, for $\lambda=-\ii b, $ we have $\ee^{-2\ii\rho(\lambda)x}=\ee^{-2\sqrt{b^2-1}x}$, also exhibiting exponential decay. 
\end{proof}
Then we can calculate the large $x$ asymptotics at $t=0$ when $x\to+\infty$ directly. Since the first two jump matrices in Eq.\eqref{eq:jump-P} approach to the identity matrix, the leading-order term of the asymptotics as $x\to+\infty$ will be determined by the third jump condition in Eq.\eqref{eq:jump-P}. By applying the Plemelj formula to the third jump condition of $\mathbf{P}(\lambda; x, 0)$, we obtain the following result, 
\begin{equation}
	\mathbf{P}(\lambda; x, 0)=\ee^{\ii\pi/4\sigma_3}\mathbf{C}\left(\frac{\lambda-\ii}{\lambda+\ii}\right)^{\sigma_3/4}\mathbf{C}^{-1}\ee^{-\ii\pi/4\sigma_3},\quad\mathbf{C}:=\frac{1}{\sqrt{2}}\begin{bmatrix}1&1\\\ii&-\ii
	\end{bmatrix}.
\end{equation}
Then the large $x$ asymptotics at $x\to+\infty$ is shown in the following theorem.
\begin{theorem}(Large $x$ asymptotics as $x\to+\infty$) For $t=0$, when $x\to+\infty$, the asymptotics of the soliton gas for the mKdV equation can be given as: 
\begin{equation}\label{eq:q}
	q(x,0)=2\ii\lim\limits_{\lambda\to\infty}	\lambda \mathbf{P}(\lambda; x, 0)_{12}=1+\mathcal{O}(\ee^{-x\mu}), (\mu>0). 
\end{equation} 
\end{theorem}
However, when $x\to-\infty$, the first two jump matrices of $\mathbf{P}(\lambda; x, t)$ are exponentially large, resulting in a significantly different behavior for the corresponding asymptotics of $q(x,0)$. Consequently, the asymptotic expression provided in Eq.\eqref{eq:q} becomes invalid for $x\to-\infty$. Therefore, we must employ a new approach to address this scenario, as detailed in the subsequent subsection. 

\subsection{The asymptotics of the potential $q(x,0)$ as $x\to-\infty$ }
Compared to the large $x$ asymptotics as $x\to+\infty$, the behavior as $x\to-\infty$ becomes significantly more complex. The method used in the last subsection is no longer applicable, as the exponential terms in the jump conditions of $\mathbf{P}(\lambda; x, 0)$ become large. Thus we need to introduce a new way  to transform the triangularity matrices into a converse form. Beforehand, we first give a  triangular decomposition for the jump matrices,
\begin{equation}\label{eq:jump-decom}
	\begin{aligned}
		&\begin{bmatrix}
			1&0\\-2\ii r_1(\lambda)&1
			\end{bmatrix}=\begin{bmatrix}1&\frac{1}{-2\ii r_1(\lambda)}\\0&1
			\end{bmatrix}\begin{bmatrix}0&\frac{1}{2\ii r_1(\lambda)}\\
			-2\ii r_1(\lambda)&0
			\end{bmatrix}\begin{bmatrix}1&\frac{1}{-2\ii r_1(\lambda)}\\0&1
		\end{bmatrix},\\
	&\begin{bmatrix}1&-2\ii r_1(\lambda)\\0&1
		\end{bmatrix}=\begin{bmatrix}1&0\\\frac{1}{-2\ii r_1(\lambda)}&1
		\end{bmatrix}\begin{bmatrix}0&-2\ii r_1(\lambda)\\
		\frac{1}{2\ii r_1(\lambda)}&0
		\end{bmatrix}\begin{bmatrix}1&0\\\frac{1}{-2\ii r_1(\lambda)}&1
	\end{bmatrix}.
		\end{aligned}
	\end{equation}
Considering the phase term $\ee^{\pm2\ii\theta(\lambda; x, t)\sigma_3}$ in the jump matrix of $\mathbf{P}(\lambda; x, t)$ contains the problematic factor $\rho(\lambda)$, which is not analytic for $\lambda\in\Sigma_c$. We define a new matrix function $\mathbf{Q}(\lambda; x, t)$ to eliminate this factor, 
\begin{equation}
	\mathbf{Q}(\lambda; x, t):=\mathbf{P}(\lambda; x, t)\ee^{-\ii\theta(\lambda; x, t)\sigma_3}\ee^{\ii\lambda(x+4\lambda^2 t)\sigma_3},
\end{equation}
then we have the following Riemann-Hilbert problem. 
\begin{rhp}Seek a $2\times 2$ matrix value function $\mathbf{Q}(\lambda; x, t)$ with the following properties. 
	\begin{itemize}
		\item 
		{\bf Analyticity:} $\mathbf{Q}(\lambda; x, t)$ is analytic for $\lambda\in\mathbb{C}\setminus \left(\left[\ii\eta_1, \ii\eta_2\right]\cup \left[-\ii\eta_1, -\ii\eta_2\right]\cup\Sigma_c\right)$.
		\item 
		{\bf Jump condition:} $\mathbf{Q}(\lambda; x, t)$ takes continuous boundary values on $\left[\ii\eta_1, \ii\eta_2\right]\cup \left[-\ii\eta_1, -\ii\eta_2\right]\cup\Sigma_c$, and they are related by the jump conditions:
		\begin{equation}\label{eq:jump-Q}
			\begin{aligned}
				&\mathbf{Q}_{+}(\lambda; x, t)=\mathbf{Q}_{-}(\lambda; x, t)\ee^{-\ii\lambda(x+4\lambda^2 t)\sigma_3}\begin{bmatrix}1&0\\-2\ii r_1(\lambda)&1
				\end{bmatrix}\ee^{\ii\lambda(x+4\lambda^2 t)\sigma_3},\quad &&\lambda\in\left[\ii\eta_1, \ii\eta_2\right],\\
				&\mathbf{Q}_{+}(\lambda; x, t)=\mathbf{Q}_{-}(\lambda; x, t)\ee^{-\ii\lambda(x+4\lambda^2 t)\sigma_3}\begin{bmatrix}1&-2\ii r_1(\lambda)\\0&1
				\end{bmatrix}\ee^{\ii\lambda(x+4\lambda^2 t)\sigma_3},\quad &&\lambda\in\left[-\ii\eta_1,-\ii\eta_2\right],\\
				&\mathbf{Q}_{+}(\lambda; x, t)=\mathbf{Q}_{-}(\lambda; x, t)\ee^{-\ii\lambda(x+4\lambda^2 t)\sigma_3}\begin{bmatrix}0&\ii\\\ii&0
				\end{bmatrix}\ee^{\ii\lambda(x+4\lambda^2 t)\sigma_3},\quad&&\lambda\in\Sigma_c.
			\end{aligned}
		\end{equation} 
		\item {\bf Normalization:} As $\lambda\to\infty$, we have $$\lim\limits_{\lambda\to\infty}\mathbf{Q}(\lambda; x, t)\to\mathbb{I}.$$
	\end{itemize}
	Then the potential $q(x,t)$ can be recovered by 
	\begin{equation}
		q(x, t)=2\ii\lim\limits_{\lambda\to\infty}\lambda\mathbf{Q}(\lambda; x, t)_{12}.
	\end{equation}
\end{rhp}
When $t=0$, the jump matrices of $\mathbf{Q}(\lambda; x, t)$ become the following case,
	\begin{equation}\label{eq:jump-Q}
	\begin{aligned}
		&\mathbf{Q}_{+}(\lambda; x, 0)=\mathbf{Q}_{-}(\lambda; x, 0)\ee^{-\ii\lambda x\sigma_3}\begin{bmatrix}1&0\\-2\ii r_1(\lambda)&1
		\end{bmatrix}\ee^{\ii\lambda x\sigma_3},\quad &&\lambda\in\left[\ii\eta_1, \ii\eta_2\right],\\
		&\mathbf{Q}_{+}(\lambda; x, 0)=\mathbf{Q}_{-}(\lambda; x, 0)\ee^{-\ii\lambda x\sigma_3}\begin{bmatrix}1&-2\ii r_1(\lambda)\\0&1
		\end{bmatrix}\ee^{\ii\lambda x\sigma_3},\quad &&\lambda\in\left[-\ii\eta_1,-\ii\eta_2\right],\\
		&\mathbf{Q}_{+}(\lambda; x, 0)=\mathbf{Q}_{-}(\lambda; x, 0)\ee^{-\ii\lambda x\sigma_3}\begin{bmatrix}0&\ii\\\ii&0
		\end{bmatrix}\ee^{\ii\lambda x\sigma_3},\quad&&\lambda\in\Sigma_c.
	\end{aligned}
\end{equation} 
In the case of large positive $x$, we compute the large $x$ asymptotics from $\mathbf{P}(\lambda; x, 0)$. While in this case, we should introduce a $g$-function to simultaneously modify these three jump matrices of $\mathbf{Q}(\lambda; x, 0)$. Subsequently, we apply the Deift-Zhou nonlinear steepest descent method to obtain the asymptotics. For brevity, we denote the intervals as follows: $\left[\ii\eta_1, \ii\eta_2\right]=\Sigma_{+}, \left[\ii, \ii\eta_1\right]=\Gamma_{+}, \left[-\ii,-\ii\eta_1\right]=\Gamma_{-}, \left[-\ii\eta_1, -\ii\eta_2\right]=\Sigma_{-}$. Then we formulate the Riemann-Hilbert problem for the $g$-function. 
\begin{rhp}($g$-function in the large negative $x$)
	\begin{itemize}
		\item {\bf Analyticity:} $g$-function is analytic for $\lambda\in\mathbb{C}\setminus\left(\Sigma_{\pm}\cup\Gamma_{\pm}\cup\Sigma_c\right)$.
		\item {\bf Jump condition:} $g$-function takes continuous values on these intervals, which are related by the following jump conditions,
		\begin{equation}
			\begin{aligned}
		&g_{+}(\lambda)+g_{-}(\lambda)=2\ii\lambda+\kappa_1,\quad&&\lambda\in\Sigma_{+},\\
		&g_{+}(\lambda)-g_{-}(\lambda)=d_1, \quad&&\lambda\in\Gamma_{+},\\
		&g_{+}(\lambda)+g_{-}(\lambda)=2\ii\lambda+\kappa_2,\quad&&\lambda\in\Sigma_c,\\
		&g_{+}(\lambda)-g_{-}(\lambda)=d_2, \quad&&\lambda\in\Gamma_{-},\\
		&g_{+}(\lambda)+g_{-}(\lambda)=2\ii\lambda+\kappa_3,\quad&&\lambda\in\Sigma_{-}.
		\end{aligned}
	\end{equation}
\item {\bf Normalization:} As $\lambda\to\infty$, $g(\lambda)$ satisfies the normalization condition, 
\begin{equation}
	g(\lambda)\to\mathcal{O}(\lambda^{-1}).
	\end{equation}
\item {\bf Symmetry:} $g(\lambda)$ has the symmetry condition: $g(\lambda)=-g(\lambda^*)^*.$
		\end{itemize}
	\end{rhp}  
With this symmetry  $g(\lambda)=-g(\lambda^*)^*$, the integral constants $\kappa_3=\kappa_1, d_2=d_1$ are all pure imaginary.

To solve this Riemann-Hilbert problem, we introduce a $R(\lambda)$ function with the branch cut at $\Sigma_{\pm}\cup\Sigma_c$, which is defined as follows: 
\begin{equation}\label{eq:R-function}
	R(\lambda):=\sqrt{(\lambda-\ii)(\lambda+\ii)(\lambda+\ii\eta_1)(\lambda-\ii\eta_1)(\lambda+\ii\eta_2)(\lambda-\ii\eta_2)}.
\end{equation}
This function $R(\lambda)$ is related to an algebraic curve of genus-two on the Riemann surface. Let $\alpha_1, \alpha_2, \beta_1, \beta_2$ be a homology basis on the Riemann surface such that they satisfy the intersection properties $\alpha_1\cdot\alpha_2=\beta_1\cdot\beta_2=\alpha_1\cdot\beta_2=\alpha_2\cdot\beta_1=0,\alpha_1\cdot\beta_1=\alpha_2\cdot\beta_2=1,$ see Fig.\ref{fig:Rie-sur}.
\begin{figure}[ht]
	\centering
	\includegraphics[width=0.3\textwidth,angle=-90]{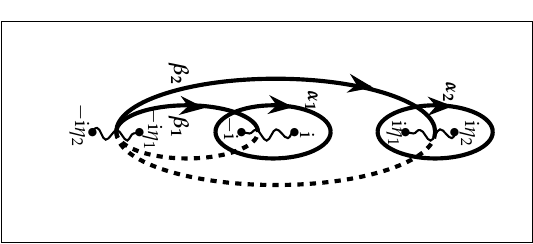}
	\caption{Homology cycles on the Riemann surface. The snake curves represent the branch cuts. The solid curves are on the first sheet and the dashed curves are on the second sheet.}
	\label{fig:Rie-sur}
\end{figure}

By differentiating this Riemann-Hilbert problem, we have 
\begin{equation}
	\begin{aligned}
		g_{+}'(\lambda)+g_{-}'(\lambda)=2\ii,\quad\lambda\in\left(\Sigma_{\pm}\cup\Sigma_c\right).
	\end{aligned}
	\end{equation} 
Utilizing the defined $R(\lambda)$, we have
\begin{equation}
	\left(g'(\lambda)R(\lambda)\right)_{+}-\left(g'(\lambda)R(\lambda)\right)_{-}=\left(2\ii R(\lambda)\right)_{+},
	\end{equation}
which can be solved by the Plemelj formula, 
\begin{equation}
	g'(\lambda)R(\lambda)=\frac{1}{2\pi\ii}\int_{\Sigma_{\pm}\cup\Sigma_c}\frac{2\ii R(\xi)}{\xi-\lambda}d\xi+\tilde{c}_1\lambda+c_0, 
	\end{equation}
where these last two added factors ensure the normalization condition as $\lambda\to\infty$. Moreover, with the generalized residue theorem, $g'(\lambda)R(\lambda)$ can be expressed as
\begin{equation}
	g'(\lambda)R(\lambda)=\ii R(\lambda)-\ii\lambda^3-\ii(\eta_1^2+\eta_2^2+1)\lambda/2+\ii\tilde{c}_1\lambda-\ii c_0:=\ii R(\lambda)-\ii\lambda^3-\ii c_1\lambda-\ii c_0,
	\end{equation}
where $c_1=(\eta_1^2+\eta_2^2+1)/2-\tilde{c}_1. $ 
Then we have
\begin{equation}
	g'(\lambda)=\ii-\frac{\ii\lambda^3+\ii c_1\lambda+\ii c_0}{R(\lambda)},
	\end{equation}
with the normalization as $\lambda\to\infty$, $g(\lambda)$ can be defined by the integration from infinity:
\begin{equation}
	g(\lambda):=\int_{\infty}^{\lambda}\left(\ii-\frac{\ii\xi^3+\ii c_1\xi+\ii c_0}{R(\xi)}\right)d\xi.
	\end{equation} 
When $\lambda\in\Sigma_{+}$, the integral constant $\kappa_1$ is calculated as 
\begin{equation}
	\kappa_1=2\int^{\ii\eta_2}_{\infty}\left(\ii-\frac{\ii\xi^3+\ii c_1\xi+\ii c_0}{R(\xi)}\right)d\xi+2\eta_2.
	\end{equation}
When $\lambda$ is in the other intervals, the other integral constants can be expressed in a similar integral form:
\begin{equation}\label{eq:inter-cons}
	\begin{aligned}
	d_1&=-2\int_{\ii\eta_2}^{\ii\eta_1}\frac{\ii\xi^3+\ii c_1\xi+\ii c_0}{R(\xi)}d\xi,\quad\quad
	\kappa_2=\kappa_1-2\int_{\ii\eta_1}^{\ii}\frac{\ii\xi^3+\ii c_1\xi+\ii c_0}{R(\xi)}d\xi,\\
	d_2&=-2\int_{-\ii\eta_2}^{-\ii\eta_1}\frac{\ii\xi^3+\ii c_1\xi+\ii c_0}{R(\xi)}d\xi, \quad\,\,\kappa_3=2\int_{\infty}^{-\ii\eta_2}\left(\ii-\frac{\ii\xi^3+\ii c_1\xi+\ii c_0}{R(\xi)}\right)d\xi-2\eta_2.
	\end{aligned}
	\end{equation}
With the constraint condition $d_2=d_1$, we know $c_0=0$. To calculate the other parameter $c_1$, we impose the integral constants $\kappa_2, \kappa_1$ as pure imaginary. Then we have  $\Re\left(\int_{\ii\eta_1}^{\ii}\frac{\ii\xi^3+\ii c_1\xi}{R(\xi)}d\xi\right)=0$. It is evident that this integral $\int_{\ii\eta_1}^{\ii}\frac{\ii\xi^3+\ii c_1\xi}{R(\xi)}d\xi$ is a real number, then we get a relation concerning the parameter $c_1$,
\begin{equation}\label{eq:cons-cond-1}
	\int_{\ii\eta_1}^{\ii}\frac{\ii\xi^3+\ii c_1\xi}{R(\xi)}d\xi=0,
	\end{equation}
thus we have
\begin{equation}\label{eq:c1}
	c_1=-\frac{\int_{\ii\eta_1}^{\ii}\frac{\xi^3}{R(\xi)}d\xi}{\int_{\ii\eta_1}^{\ii}\frac{\xi}{R(\xi)}d\xi}.
	\end{equation}
Moreover, we find that these integral constants $\kappa_1=\kappa_2=\kappa_3$ possess some special characteristic, as illustrated in the following proposition. 
\begin{prop}\label{prop:inter-cons}
The integral constants $\kappa_1=\kappa_2=\kappa_3$ are all zero.
\end{prop}
\begin{proof}
Combining Eq.\eqref{eq:cons-cond-1}, Eq.\eqref{eq:inter-cons} and the symmetry condition $g(\lambda)=-g(\lambda^*)^*$, we deduce that $\kappa_3=\kappa_2=\kappa_1$, thus we only need to prove that one of them is zero. Let us take $\kappa_3$ as an example to demonstrate this. For brevity, we rewrite the integral $\kappa_3$ as a limit formula:
\begin{equation}
\kappa_3=-2\lim\limits_{\lambda\to\infty}\int_{\lambda}^{-\ii\eta_2}\frac{\ii\xi^3+\ii c_1\xi}{R(\xi)}d\xi-2\ii\lambda.
\end{equation}
Suppose functions $\Omega$ and $\widehat{\Omega}$ are holomorphic on the Riemann surface $\pmb{\chi}$. According to the Riemann bilinear relation, we have 
\begin{equation}\label{eq:Rie-bili}
0=\int_{\pmb{\chi}}d\left(\Omega d\widehat{\Omega}\right)=\oint_{\partial \pmb{\chi}}\Omega d\widehat{\Omega}=\sum\limits_{j=1}^{2}(A_j\widehat{B}_{j}-B_j\widehat{A}_{j}),
\end{equation}
where $A_j=\oint_{\alpha_j}d\Omega,\widehat{A}_j=\oint_{\alpha_j}d\widehat{\Omega}, B_j=\oint_{\beta_j}d\Omega, \widehat{B}_j=\oint_{\beta_j}d\widehat{\Omega}$. In fact, we want to use a differential with singularities at $\lambda\to\infty$ to prove $\kappa_3=0$. Consequently, the Riemann bilinear formula Eq.\eqref{eq:Rie-bili} should be modified by adding an integral $\Omega d\widehat{\Omega}$ along a clockwise path at the singularities $\lambda\to\infty$, that is
\begin{equation}\label{eq:Rie-bili-2}
0=\sum\limits_{j=1}^{2}(A_j\widehat{B}_{j}-B_j\widehat{A}_{j})+\oint_{\ell+\iota(\ell)}\Omega d\widehat{\Omega}.
\end{equation}
Take $\Omega=-2\int_{\lambda}^{-\ii\eta_2}\frac{\ii\zeta^3+\ii c_1\zeta}{R(\zeta)}d\zeta, d\widehat{\Omega}=\frac{\lambda^2}{R(\lambda)}d\lambda$. Expanding $\Omega d\widehat{\Omega}$ at $\lambda\to\infty$, we have 
\begin{equation}
\Omega d\widehat{\Omega}=\left(2\ii+\kappa_3\lambda^{-1}+\mathcal{O}(\lambda^{-2})\right)d\lambda.
\end{equation}
Using the Riemann bilinear relation Eq.\eqref{eq:Rie-bili-2}, we have
\begin{multline}\label{eq:cal-kappa3}
2\oint_{\alpha_1}\frac{\ii\lambda^3+\ii c_1\lambda}{R(\lambda)}d\lambda \oint_{\beta_1}\frac{\lambda^2}{R(\lambda)}d\lambda-2\oint_{\beta_1}\frac{\ii\lambda^3+\ii c_1\lambda}{R(\lambda)}d\lambda \oint_{\alpha_1}\frac{\lambda^2}{R(\lambda)}d\lambda+2\oint_{\alpha_2}\frac{\ii\lambda^3+\ii c_1\lambda}{R(\lambda)}d\lambda \oint_{\beta_2}\frac{\lambda^2}{R(\lambda)}d\lambda\\-2\oint_{\beta_2}\frac{\ii\lambda^3+\ii c_1\lambda}{R(\lambda)}d\lambda \oint_{\alpha_2}\frac{\lambda^2}{R(\lambda)}d\lambda+4\pi\ii\kappa_3=0.
\end{multline}
From Eq.\eqref{eq:cons-cond-1}, we know that $\oint_{\beta_1}\frac{\ii\lambda^3+\ii c_1\lambda}{R(\lambda)}d\lambda=\oint_{\beta_2}\frac{\ii\lambda^3+\ii c_1\lambda}{R(\lambda)}d\lambda=0$. Then Eq.\eqref{eq:cal-kappa3} is simplified into
\begin{equation}
2\oint_{\alpha_1}\frac{\ii\lambda^3+\ii c_1\lambda}{R(\lambda)}d\lambda \oint_{\beta_1}\frac{\lambda^2}{R(\lambda)}d\lambda+2\oint_{\alpha_2}\frac{\ii\lambda^3+\ii c_1\lambda}{R(\lambda)}d\lambda \oint_{\beta_2}\frac{\lambda^2}{R(\lambda)}d\lambda+4\pi\ii\kappa_3=0.
\end{equation}
Moreover, with the symmetry relation of homology cycles on Riemann surface (See Fig.\ref{fig:Rie-sur}), the first integral $\oint_{\alpha_1}\frac{\ii\lambda^3+\ii c_1\lambda}{R(\lambda)}d\lambda$ and the second integral $\oint_{\beta_2}\frac{\lambda^2}{R(\lambda)}d\lambda$ are both zero. Therefore, we conclude that the integral constant $\kappa_3$ is zero. 
\end{proof}

With this $g$-function, we begin to deform the jump matrices to separate the exponential factor with the triangular decomposition in Eq.\eqref{eq:jump-decom}. Let us define 
\begin{equation}
\mathbf{T}(\lambda; x, 0):=\left\{\begin{aligned}
	&f_{\infty}^{-\sigma_3}\mathbf{Q}(\lambda; x, 0)\ee^{-\ii\lambda x\sigma_3}\begin{bmatrix}1&\frac{1}{-2\ii r_1(\lambda)}\\0&1
		\end{bmatrix}\ee^{\ii\lambda x\sigma_3}\ee^{-xg(\lambda)\sigma_3}f(\lambda)^{\sigma_3},\quad&&\lambda\in R_{+},\\
	&f_{\infty}^{-\sigma_3}\mathbf{Q}(\lambda; x, 0)\ee^{-\ii\lambda x\sigma_3}\begin{bmatrix}1&\frac{1}{2\ii r_1(\lambda)}\\
	0&1\end{bmatrix}\ee^{\ii\lambda x\sigma_3}\ee^{-xg(\lambda)\sigma_3}f(\lambda)^{\sigma_3},\quad&&\lambda\in L_{+},\\
&f_{\infty}^{-\sigma_3}\mathbf{Q}(\lambda; x, 0)\ee^{-\ii\lambda x\sigma_3}\begin{bmatrix}1&0\\\frac{1}{-2\ii r_1(\lambda)}&1
\end{bmatrix}\ee^{\ii\lambda x\sigma_3}\ee^{-xg(\lambda)\sigma_3}f(\lambda)^{\sigma_3},\quad&&\lambda\in R_{-},\\
	&f_{\infty}^{-\sigma_3}\mathbf{Q}(\lambda; x, 0)\ee^{-\ii\lambda x\sigma_3}\begin{bmatrix}1&0\\\frac{1}{2\ii r_1(\lambda)}&
1\end{bmatrix}\ee^{\ii\lambda x\sigma_3}\ee^{-xg(\lambda)\sigma_3}f(\lambda)^{\sigma_3},\quad&&\lambda\in L_{-},\\
&f_{\infty}^{-\sigma_3}\mathbf{Q}(\lambda; x, 0)\ee^{-xg(\lambda)\sigma_3}f(\lambda)^{\sigma_3},\quad&&\text{otherwise}.
	\end{aligned}\right.
	\end{equation}
The new factor $f(\lambda)$ is introduced to eliminate $r_{1}(\lambda)$, satisfying the following Riemann-Hilbert problem, 
\begin{rhp}(The Riemann-Hilbert problem of $f(\lambda)$)
	\begin{itemize}
		\item {\bf Analyticity:} $f(\lambda)$ is analytic for $\lambda\in\mathbb{C}\setminus \left(\Sigma_{\pm}\cup \Gamma_{\pm}\cup\Sigma_c\right)$.
		\item {\bf Jump condition:} For $\lambda\in\left(\Sigma_{\pm}\cup \Gamma_{\pm}\cup\Sigma_c\right),$ $f_{\pm}(\lambda)$ are related to the following jump conditions:
		\begin{equation}
			\begin{aligned}
				f_{+}(\lambda)f_{-}(\lambda)&=\frac{1}{2\ii r_1(\lambda)},\quad&&\lambda\in\Sigma_{+},\\
				f_{+}(\lambda)f_{-}(\lambda)&=\ii,\quad&&\lambda\in\Sigma_{c},\\
				f_{+}(\lambda)f_{-}(\lambda)&=-2\ii r_1(\lambda), \quad&&\lambda\in\Sigma_{-},\\
				\frac{f_{+}(\lambda)}{f_{-}(\lambda)}&=\ee^{\Delta_1},\quad&&\lambda\in\Gamma_{+},\\
				\frac{f_{+}(\lambda)}{f_{-}(\lambda)}&=\ee^{\Delta_2},\quad&&\lambda\in\Gamma_{-}.
				\end{aligned}
			\end{equation}
		\item {\bf Normalization:} As $\lambda\to\infty$, we have $f(\lambda)\to f_{\infty}+\mathcal{O}(\lambda^{-1})$.
		\end{itemize}
	\end{rhp}
To solve this Riemann-Hilbert problem, we take the logarithm of the above $f(\lambda)$ function, that is
\begin{equation}
	\begin{aligned}
		\log\left(f_{+}(\lambda)\right)+\log\left(f_{-}(\lambda)\right)&=-\log\left(2\ii r_1(\lambda)\right),&&\quad\lambda\in\Sigma_{+},\\
		\log\left(f_{+}(\lambda)\right)+\log\left(f_{-}(\lambda)\right)&=\log(\ii)=\frac{\pi}{2}\ii,&&\quad \lambda\in\Sigma_c,\\
		\log\left(f_{+}(\lambda)\right)+\log\left(f_{-}(\lambda)\right)&=\log(-2\ii r_1(\lambda)),&&\quad\lambda\in\Sigma_{-},\\
		\log\left(f_{+}(\lambda)\right)-\log\left(f_{-}(\lambda)\right)&=\Delta_1,&&\quad\lambda\in \Gamma_{+},\\
		\log\left(f_{+}(\lambda)\right)-\log\left(f_{-}(\lambda)\right)&=\Delta_2,&&\quad\lambda\in\Gamma_{-},
		\end{aligned}
	\end{equation}
by using the function $R(\lambda)$, the expression $\log\left(f(\lambda)\right)$ can be solved through the Plemelj formula, as follows:  
\begin{multline}
	\log\left(f(\lambda)\right)=\frac{R(\lambda)}{2\pi\ii}\Bigg(\int_{\Sigma_{+}}-\frac{\log(2\ii r_1(\xi))}{R(\xi)(\xi-\lambda)}d\xi+\int_{\Sigma_{c}}\frac{\frac{\pi}{2}\ii}{R(\xi)(\xi-\lambda)}d\xi+\int_{\Sigma_{-}}\frac{\log(-2\ii r_1(\xi))}{R(\xi)(\xi-\lambda)}d\xi\\+\int_{\Gamma
_{+}}\frac{\Delta_1}{R(\xi)(\xi-\lambda)}d\xi+\int_{\Gamma_{-}}\frac{\Delta_2}{R(\xi)(\xi-\lambda)}d\xi\Bigg).
	\end{multline}
The normalization condition $f(\lambda)\to f_{\infty}+\mathcal{O}(\lambda^{-1})$ leads to the derivation of two equations regarding $\Delta_1$ and $\Delta_2$, 
\begin{equation}
\begin{aligned}
	&\int_{\Sigma_{+}}-\frac{\log(2\ii r_1(\xi))}{R(\xi)}d\xi+\int_{\Sigma_{c}}\frac{\frac{\pi}{2}\ii}{R(\xi)}d\xi+\int_{\Sigma_{-}}\frac{\log(-2\ii r_1(\xi))}{R(\xi)}d\xi+\int_{\Gamma
		_{+}}\frac{\Delta_1}{R(\xi)}d\xi+\int_{\Gamma_{-}}\frac{\Delta_2}{R(\xi)}d\xi=0,\\
	&\int_{\Sigma_{+}}-\frac{\log(2\ii r_1(\xi))}{R(\xi)}\xi d\xi+\int_{\Sigma_{c}}\frac{\frac{\pi}{2}\ii}{R(\xi)}\xi d\xi+\int_{\Sigma_{-}}\frac{\log(-2\ii r_1(\xi))}{R(\xi)}\xi d\xi+\int_{\Gamma
		_{+}}\frac{\Delta_1}{R(\xi)}\xi d\xi+\int_{\Gamma_{-}}\frac{\Delta_2}{R(\xi)}\xi d\xi=0,
	\end{aligned}
	\end{equation}
which can uniquely determine the unknown parameters $\Delta_1, \Delta_2$. As $\lambda\to\infty$, $f_{\infty}$ can be represented in the following integral form, 
\begin{equation}\label{eq:finfinity}
f_{\infty}=\ee^{\frac{1}{2\pi\ii}\left(\int_{\Sigma_{+}}\frac{\log(2\ii r_1(\xi))}{R(\xi)}\xi^2 d\xi-\int_{\Sigma_{c}}\frac{\frac{\pi}{2}\ii}{R(\xi)}\xi^2 d\xi-\int_{\Sigma_{-}}\frac{\log(-2\ii r_1(\xi))}{R(\xi)}\xi^2 d\xi-\int_{\Gamma
		_{+}}\frac{\Delta_1}{R(\xi)}\xi^2 d\xi-\int_{\Gamma_{-}}\frac{\Delta_2}{R(\xi)}\xi^2 d\xi\right)}.
\end{equation}
Then the jump conditions of  $\mathbf{T}(\lambda; x, 0)$ can be rewritten as another equivalent forms, 
\begin{equation}\label{eq:jump-T}
	\begin{aligned}
	&\mathbf{T}_{+}(\lambda; x, 0)=\mathbf{T}_{-}(\lambda; x, 0)\begin{bmatrix}
		1&-\frac{1}{2\ii r_1(\lambda)}\ee^{-2\ii\lambda x+2xg(\lambda)}\frac{1}{f^2}\\
		0&1
	\end{bmatrix},&&\quad\lambda\in C_{L}^{+},\\
&\mathbf{T}_{+}(\lambda; x, 0)=\mathbf{T}_{-}(\lambda; x, 0)\begin{bmatrix}0&1\\
	-1&0
	\end{bmatrix}, &&\quad\lambda\in\Sigma_{+},\\
&\mathbf{T}_{+}(\lambda; x, 0)=\mathbf{T}_{-}(\lambda; x, 0)\begin{bmatrix}1&-\frac{1}{2\ii r_1(\lambda)}\ee^{-2\ii\lambda x+2xg(\lambda)}\frac{1}{f^2}\\
	0&1
	\end{bmatrix},&&\quad\lambda\in C_{R}^{+},\\
&\mathbf{T}_{+}(\lambda; x, 0)=\mathbf{T}_{-}(\lambda; x, 0)\begin{bmatrix}\ee^{-xd_1+\Delta_1}&0\\0&\ee^{xd_1-\Delta_1}
	\end{bmatrix},&&\quad\lambda\in\Gamma_{+},\\
&\mathbf{T}_{+}(\lambda; x, 0)=\mathbf{T}_{-}(\lambda; x, 0)\begin{bmatrix}
0&1\\
-1&0	\end{bmatrix},&&\quad\lambda\in\Sigma_c,\\
&\mathbf{T}_{+}(\lambda; x, 0)=\mathbf{T}_{-}(\lambda; x, 0)\begin{bmatrix}\ee^{-xd_1+\Delta_2}&0\\
	0&\ee^{xd_1-\Delta_2}
	\end{bmatrix},&&\quad\lambda\in\Gamma_{-},\\
&\mathbf{T}_{+}(\lambda; x, 0)=\mathbf{T}_{-}(\lambda; x, 0)\begin{bmatrix}
	1&0\\\frac{1}{-2\ii r_1(\lambda)}\ee^{2\ii\lambda x-2xg(\lambda)}f^2&1
\end{bmatrix},&&\quad\lambda\in C_{L}^{-},\\
&\mathbf{T}_{+}(\lambda; x, 0)=\mathbf{T}_{-}(\lambda; x, 0)\begin{bmatrix}
	1&0\\\frac{1}{-2\ii r_1(\lambda)}\ee^{2\ii\lambda x-2xg(\lambda)}f^2&1
\end{bmatrix},&&\quad\lambda\in C_{R}^{-},\\
&\mathbf{T}_{+}(\lambda; x, 0)=\mathbf{T}_{-}(\lambda; x, 0)\begin{bmatrix}0&1\\
	-1&0
\end{bmatrix}, &&\quad\lambda\in\Sigma_{-}.
\end{aligned}
	\end{equation}
The jump contours of $\mathbf{T}(\lambda; x, 0)$ are illustrated in Fig.\ref{fig:jump-T}.
\begin{figure}[ht]
	\centering
	\includegraphics[width=0.5\textwidth]{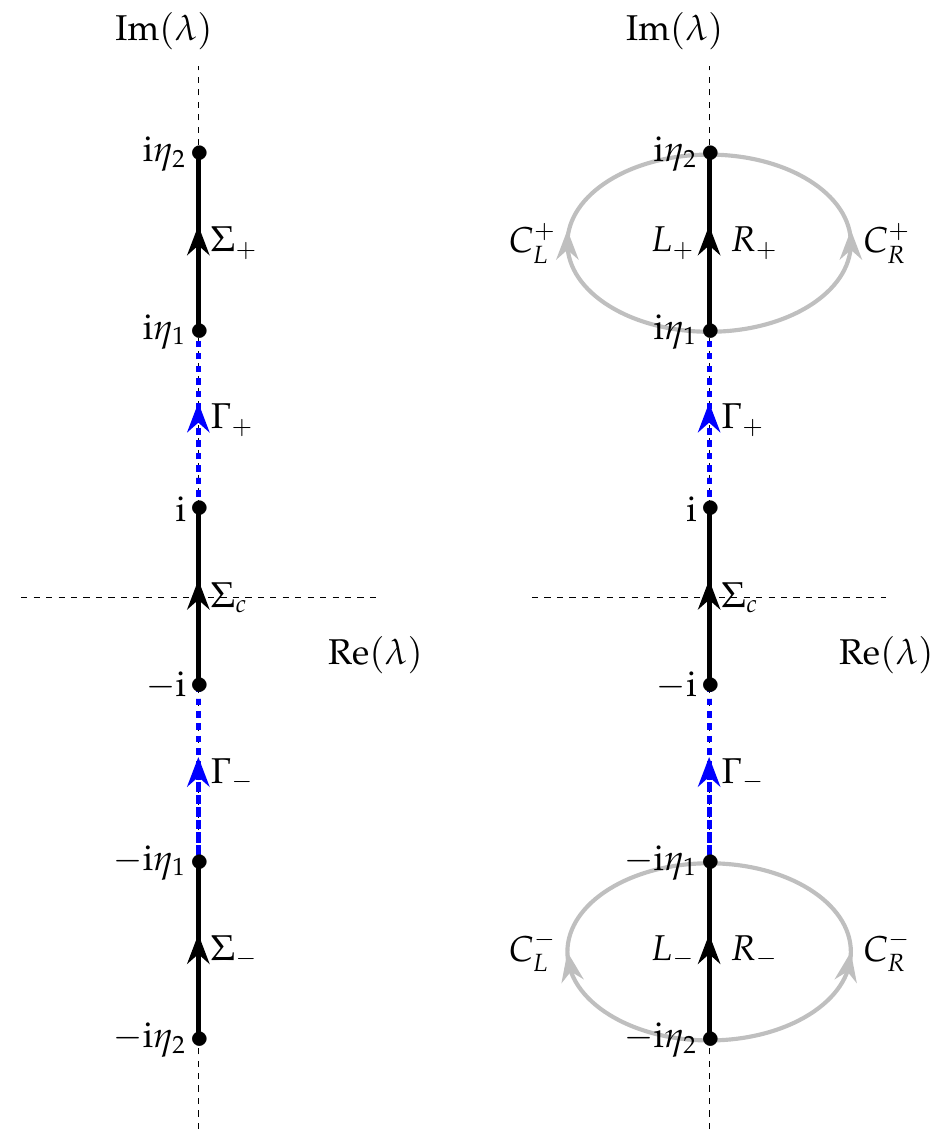}
	\caption{The left panel is the jump contours of the $g$-function and the right one is the jump contours of $\mathbf{T}(\lambda; x, 0)$. }
	\label{fig:jump-T}
\end{figure}
For large negative $x$, the exponential terms of $\mathbf{T}(\lambda; x, 0)$ will decay to zero in these corresponding regions, as demonstrated by Lemma \ref{lemma:2}. 
\begin{lemma}\label{lemma:2}
In the contours $C_{L}^{+}, C_{R}^{+}$, $\Re(-2\ii\lambda+2g(\lambda))>0$. Conversely, in the contours $C_{L}^{-}, C_{R}^{-},\Re(-2\ii\lambda+2g(\lambda))<0$. 
\end{lemma}
\begin{proof}
For brevity, let us denote the complex variable $\lambda$ as $\lambda=x+\ii y$ and the function $-2\ii\lambda+2g(\lambda)=u(x, y)+\ii v(x,y)$. When $\lambda\in \Sigma_{+}$, we have $\Re(-2\ii\lambda+2g(\lambda))=0.$ By differentiating this function with respect to $\lambda$, when $\lambda\in\Sigma_{+}$, we have 
\begin{equation}\label{eq:ux}
u_{x}=\Re\left(-2\ii+2g'_{+}(\lambda)\right)=\Re\left(-\frac{\ii\lambda^3+\ii c_1\lambda}{R_{+}(\lambda)}\right),
\end{equation} 
from the definition of $c_1$ in Eq.\eqref{eq:c1}, we know that $c_1$ is a positive real number satisfying the inequality $1<c_1<\eta_1^2$. Then Eq.\eqref{eq:ux} can be simplified into
\begin{equation}
u_{x}=\frac{-\Im(\lambda)^3-c_1\Im(\lambda)}{R_{+}(\lambda)}<0.
\end{equation}
Thus when $\lambda\in C_{L}^{+}$, we have $u>0$, implying $\Re\left(-2\ii\lambda+2g(\lambda)\right)>0$. Similarly, for $\lambda\in \left(C_{R}^{+}\cup C_{L}^{-}\cup C_{R}^{-}\right)$, we can also obtain the same results. 
\end{proof}
With this lemma, we know that the jump matrices of $\mathbf{T}(\lambda; x, 0)$ will decay to the identity matrix exponentially when $\lambda\in C_{L}^{\pm}\cup C_{R}^{\pm}$. Subsequently, leveraging this decay behavior, we can construct the parametrix for $\mathbf{T}(\lambda; x, 0)$. 

\emph {The Parametrix construction of $\mathbf{T}(\lambda; x, 0)$.} Based on the jump conditions of $\mathbf{T}(\lambda; x, 0)$ in Eq.\eqref{eq:jump-T}, we initiate the construction of an outer parametrix denoted by $\dot{\mathbf{T}}^{\rm out}(\lambda; x, 0)$. This parametrix is designed to precisely satisfy all the jump conditions for $\lambda\in\left(\Sigma_{\pm}\cup\Sigma_{c}\cup\Gamma_{\pm}\right).$ This process closely follows our previous article studying the large-order asymptotics of KM breathers\cite{ling2022large}. Then we give the Riemann-Hilbert problem for this outer parametrix $\dot{\mathbf{T}}^{\rm out}(\lambda; x, 0)$.
\begin{rhp}(Riemann-Hilbert problem of the outer parametrix $\dot{\mathbf{T}}^{\rm out}(\lambda; x, 0)$)
\begin{itemize}
\item {\bf Analyticity:} $\dot{\mathbf{T}}^{\rm out}(\lambda; x, 0)$ is analytic for $\lambda\in\mathbb{C}\setminus\left(\Sigma_{\pm}\cup\Sigma_c\cup\Gamma_{\pm}\right)$.
\item {\bf Jump condition:} The boundary conditions on these contours $\Sigma_{\pm}\cup\Sigma_c\cup\Gamma_{\pm}$ are related by $\dot{\mathbf{T}}^{\rm out}_{+}(\lambda; x, 0)=\dot{\mathbf{T}}^{\rm out}_{-}(\lambda; x, 0)\mathbf{V}_{\dot{\mathbf{T}}^{\rm out}}(\lambda; x, 0)$, where 
    \begin{equation}
    \mathbf{V}_{\dot{\mathbf{T}}^{\rm out}}(\lambda; x, 0)=\left\{\begin{aligned}&\begin{bmatrix}0&1\\
    -1&0
    \end{bmatrix},&&\quad \lambda\in\Sigma_{\pm}\cup\Sigma_c,\\
    &\begin{bmatrix}\ee^{-xd_1+\Delta_1}&0\\
    0&\ee^{xd_1-\Delta_1}
    \end{bmatrix},&&\quad\lambda\in\Gamma_{+},\\
    &\begin{bmatrix}\ee^{-xd_2+\Delta_2}&0\\
    0&\ee^{xd_2-\Delta_2}
    \end{bmatrix},&&\quad\lambda\in\Gamma_{-}.
    \end{aligned}\right.
    \end{equation}
\item{\bf Normalization:} $\dot{\mathbf{T}}^{\rm out}(\lambda; x, 0)\to\mathbb{I}$ as $\lambda\to\infty$. 
\end{itemize}
\end{rhp}
To solve this outer parametrix, we first introduce a scalar function to eliminate the jump conditions for $\lambda\in\Gamma_{\pm}$. That is, 
\begin{equation}
\begin{aligned}
&F_{+}(\lambda)+F_{-}(\lambda)=0,&&\quad\lambda\in\Sigma_{\pm}\cup\Sigma_c,\\
&F_{+}(\lambda)-F_{-}(\lambda)=-xd_1+\Delta_1,&&\quad\lambda\in\Gamma_{+},\\
&F_{+}(\lambda)-F_{-}(\lambda)=-xd_2+\Delta_2,&&\quad\lambda\in\Gamma_{-}.
\end{aligned}
\end{equation}
Through the Plemelj formula, we can express this scalar function $F(\lambda)$ using an integral formula:
\begin{equation}
F(\lambda)=\frac{R(\lambda)}{2\pi\ii}\left(\int_{\Gamma_{+}}\frac{-xd_1+\Delta_1}{R(\xi)(\xi-\lambda)}d\xi+\int_{\Gamma_{-}}\frac{-xd_2+\Delta_2}{R(\xi)(\xi-\lambda)}d\xi\right).
\end{equation}
When $\lambda\to\infty$, $F(\lambda)$ has a series expansion formula:
\begin{equation}
F(\lambda)=F_{2}\lambda^2+F_{1}\lambda+F_{0}+\mathcal{O}(\lambda^{-1}),
\end{equation}
where 
\begin{equation}\label{eq:F2F1F0}
\begin{aligned}
F_{2}&=-\frac{1}{2\pi\ii}\left(\int_{\Gamma_{+}}\frac{-xd_1+\Delta_1}{R(\xi)}d\xi+\int_{\Gamma_{-}}\frac{-xd_2+\Delta_2}{R(\xi)}d\xi\right),\\
F_{1}&=-\frac{1}{2\pi\ii}\left(\int_{\Gamma_{+}}\frac{-xd_1+\Delta_1}{R(\xi)}\xi d\xi+\int_{\Gamma_{-}}\frac{-xd_2+\Delta_2}{R(\xi)}\xi d\xi\right)\\
F_{0}&=-\frac{1}{2\pi\ii}\left(\int_{\Gamma_{+}}\frac{-xd_1+\Delta_1}{R(\xi)}\xi^2d\xi+\int_{\Gamma_{-}}\frac{-xd_2+\Delta_2}{R(\xi)}\xi^2d\xi\right)+\frac{1+\eta_1^2+\eta_2^2}{2}F_{2}.
\end{aligned}
\end{equation}
The second step for solving it is introducing an auxiliary matrix function $\mathbf{O}(\lambda; x, 0)$ by using this newly defined scalar function, 
\begin{equation}\label{eq:O-matrix}
\mathbf{O}(\lambda; x, 0)={\rm diag}\left(\ee^{F_0}, \ee^{-F_0}\right)\dot{\mathbf{T}}^{\rm out}(\lambda; x, 0){\rm diag}\left(\ee^{-F(\lambda)},\ee^{F(\lambda)}\right).
\end{equation}
It is straightforward to see that the jump conditions of $\mathbf{O}(\lambda; x, 0)$ are independent on the variable $x$,
\begin{equation}
\mathbf{O}_{+}(\lambda; x, 0)=\mathbf{O}_{-}(\lambda; x, 0)\begin{bmatrix}0&1\\-1&0\end{bmatrix}.
\end{equation}
And when $\lambda\to\infty$, $\mathbf{O}(\lambda; x, 0)$ satisfies the following series expansion,
\begin{equation}
\mathbf{O}(\lambda; x, 0)\to{\rm diag}\left(\ee^{-F_1\lambda-F_2\lambda^2}, \ee^{F_1\lambda+F_2\lambda^2}\right).
\end{equation}
Finally, we want to solve $\mathbf{O}(\lambda; x, 0)$ with the given jump conditions and the boundary condition. Before proceeding, we give a brief review of the properties of Riemann-Theta function.
\begin{definition}
The $\Theta$-function is defined as:
\begin{equation}
\Theta(u)\equiv\Theta(u; 
\mathbf{B}):=\sum\limits_{\mathbf{m}\in\mathbb{Z}^g}\ee^{\frac{1}{2}\langle\mathbf{m},\mathbf{Bm}\rangle+\langle\mathbf{m},u\rangle}.
\end{equation}
\end{definition}
This $\Theta(u)$ function has the following periodic properties,
\begin{equation}
\Theta(u+2\pi\ii\mathbf{e}_j)=\Theta(u),\quad\Theta(u+\mathbf{B}\mathbf{e}_j)=\ee^{-\frac{1}{2}B_{jj}-u_j}\Theta(u),
\end{equation}
where $\mathbf{e}_js$ are the unit vectors with the Dirac coordinate $\left(\mathbf{e}_j\right)_k=\delta_{jk},$ $\mathbf{Be}_js$ denote the $j$-column of $\mathbf{B}$.  

Next, let $\omega_j(\lambda)(j=1,2)$ be the Abel integrals \cite{belokolos1994algebro} defined as:
\begin{equation}
\omega_j(\lambda):=\int_{-\ii\eta_2}^{\lambda}\psi_{j}(\xi)d\xi, \quad \psi_{j}(\xi):=\frac{\sum\limits_{i=1}^2c_{ji}\xi^{2-\ii}}{R(\xi)},
\end{equation}
the coefficients $c_{ji}$ can be determined by the normalization condition,
\begin{equation}
\int_{\alpha_i}d\omega_{j}(\mathcal{P})=2\pi\ii\delta_{ij}, (i,j=1,2),
\end{equation}
where $d\omega_{j}(\mathcal{P})$ are holomorphic differentials on the Riemann surface. Correspondingly, the entries of the period $\mathbf{B}$ matrix can be introduced as
\begin{equation}
B_{ji}=\int_{\beta_i}d\omega_{j}(\mathcal{P}).
\end{equation}
Then we further define an Abel mapping $\mathbf{A}$ from the Riemann surface to the Jacobian variety, 
\begin{equation}\label{eq:Aj}
A_{j}(\mathcal{P})=\int_{\mathcal{P}_{0}}^{\mathcal{P}}d\omega_j(\mathcal{Q}), \quad j=1,2,
\end{equation} 
where the point $\mathcal{P}_0$ is related to the base point $-\ii\eta_2$ via the hyperelliptic projection $\pi(\mathcal{P}_{0})=-\ii\eta_2$, and $\mathcal{Q}$ is the integration variable. For two integral divisors $\mathcal{D}=\mathcal{P}_1+\mathcal{P}_2$, the Abel mapping holds right \cite{kotlyarov2017planar},
\begin{equation}
\mathbf{A}(\mathcal{D})=\mathbf{A}(\mathcal{P}_1)+\mathbf{A}(\mathcal{P}_2).
\end{equation} 
From the definition of $A_j$ in Eq.\eqref{eq:Aj}, we have 
\begin{equation}
\begin{aligned}
&\mathbf{A}_{+}(\lambda)-\mathbf{A}_{-}(\lambda)=\mathbf{0} \quad\,\,{\rm mod}\,\, 2\pi\ii\mathbb{Z}^2, \quad\lambda\in\left[-\ii\eta_1, -\ii\right]\cup\left[\ii,\ii\eta_1\right],\\
&\mathbf{A}_{+}(\lambda)+\mathbf{A}_{-}(\lambda)=\mathbf{Be}_1 \,\,{\rm mod}\,\, 2\pi\ii\mathbb{Z}^2, \quad\lambda\in\left[-\ii, \ii\right],\\
&\mathbf{A}_{+}(\lambda)+\mathbf{A}_{-}(\lambda)=\mathbf{Be}_2 \,\,{\rm mod}\,\, 2\pi\ii\mathbb{Z}^2, \quad\lambda\in\left[\ii\eta_1, \ii\eta_2\right].\\
\end{aligned}
\end{equation}
Then we also introduce another type of Abel integrals with singularities at $\mathcal{P}_{\infty_{+}}$\cite{belokolos1994algebro},
\begin{equation}
\Omega_j(\lambda)=\int_{-\ii\eta_2}^{\lambda}\Psi_{j}(\xi)d\xi\quad (j=1,2), \quad\Psi_{j}(\xi)=\frac{\sum\limits_{i=1}^5s_{ji}\xi^{5-i}}{R(\xi)},
\end{equation}
where the parameters $s_{ji}$ are determined by the normalization condition,
\begin{equation}\label{eq:nor-condi}
\int_{\alpha_i}d\Omega_{j}(\mathcal{P})=0(i, j=1,2),\quad \Omega_1(\lambda)\to\lambda+\mathcal{O}(1),\quad\Omega_{2}(\lambda)\to\lambda^2+\mathcal{O}(1),\quad \lambda\to\infty_{+}.
\end{equation} 
By integrating around with the $\beta$ circles, the entries are set as:
\begin{equation}\label{eq:UjVj}
\mathcal{U}_j=\int_{\beta_j}d\Omega_1(\mathcal{P}),\quad \mathcal{V}_j=\int_{\beta_j}d\Omega_2(\mathcal{P}),\quad j=1,2.
\end{equation} 
From the asymptotic properties of $\Omega_1(\lambda), \Omega_2(\lambda)$ as $\lambda\to\infty_{+}$, we introduce two variables $J_1, J_2$ such that the following limits exist,
\begin{equation}\label{eq:J1J2}
J_1:=\lim\limits_{\lambda\to\infty}\int_{-\ii\eta_2}^{\lambda} d\Omega_1(\mathcal{P})-\lambda,\quad J_2:=\lim\limits_{\lambda\to\infty}\int_{-\ii\eta_2}^{\lambda} d\Omega_2(\mathcal{P})-\lambda^2.
\end{equation}
Then the solution of Eq.\eqref{eq:O-matrix} can be given via the Riemann-Theta function,
\begin{multline}\label{eq:O-matrix-1}
\mathbf{O}(\lambda; x, 0)=\frac{1}{2}{\rm diag}\left(C_1, C_2\right)\begin{bmatrix}\left(\gamma(\lambda)+\frac{1}{\gamma(\lambda)}\right)\frac{\Theta\left(\mathbf{A}(\lambda)+\mathbf{d}-\pmb{\mathcal{U}}F_1-\pmb{\mathcal{V}}F_2\right)}{\Theta\left(\mathbf{A}(\lambda)+\mathbf{d}\right)}&
\ii\left(\gamma(\lambda)-\frac{1}{\gamma(\lambda)}\right)\frac{\Theta\left(\mathbf{A}(\lambda)-\mathbf{d}+\pmb{\mathcal{U}}F_1+\pmb{\mathcal{V}}F_2\right)}{\Theta\left(\mathbf{A}(\lambda)-\mathbf{d}\right)}\\
-\ii\left(\gamma(\lambda)-\frac{1}{\gamma(\lambda)}\right)\frac{\Theta\left(\mathbf{A}(\lambda)-\mathbf{d}-\pmb{\mathcal{U}}F_1-\pmb{\mathcal{V}}F_2\right)}{\Theta\left(\mathbf{A}(\lambda)-\mathbf{d}\right)}&
\left(\gamma(\lambda)+\frac{1}{\gamma(\lambda)}\right)\frac{\Theta\left(\mathbf{A}(\lambda)+\mathbf{d}+\pmb{\mathcal{U}}F_1+\pmb{\mathcal{V}}F_2\right)}{\Theta\left(\mathbf{A}(\lambda)+\mathbf{d}\right)}\\
\end{bmatrix}\\\times\ee^{-\left(\Omega_1(\lambda)F_1+\Omega_2(\lambda)F_2\right)\sigma_3},
\end{multline}
where 
\begin{equation}
\begin{aligned}
C_1&=\frac{\Theta\left(\mathbf{A}(\infty)+\mathbf{d}\right)}{\Theta\left(\mathbf{A}(\infty)+\mathbf{d}-\pmb{\mathcal{U}}F_1-\pmb{\mathcal{V}}F_2\right)}\ee^{J_1F_1+J_2F_2},\\
C_2&=\frac{\Theta\left(\mathbf{A}(\infty)+\mathbf{d}\right)}{\Theta\left(\mathbf{A}(\infty)+\mathbf{d}+\pmb{\mathcal{U}}F_1+\pmb{\mathcal{V}}F_2\right)}\ee^{-J_1F_1-J_2F_2},\\
\gamma(\lambda)&=\left(\frac{\left(\lambda+\ii\eta_2\right)\left(\lambda+\ii\eta_1\right)\left(\lambda+\ii\right)}{\left(\lambda-\ii\eta_2\right)\left(\lambda-\ii\eta_2\right)\left(\lambda+\ii\right)}\right).
\end{aligned}
\end{equation}
For $\lambda\in\Sigma_{\pm}\cup\Sigma_c, $ $\gamma(\lambda)_{+}=-\ii\gamma(\lambda)_{-}$, the vector $\mathbf{d}$ is given by 
\begin{equation}\label{eq:d}
\mathbf{d}=\mathbf{A}(\mathcal{P}_1)+\mathbf{A}(\mathcal{P}_2)+\mathbf{K},
\end{equation}
where $\mathcal{P}_1, \mathcal{P}_2$ are two zeros of $\gamma(\lambda)-\frac{1}{\gamma(\lambda)}$ on the first sheet of the Riemann surface, and $\mathbf{K}$ is the constant vector, whose entries can be simplified as\cite{belokolos1994algebro}:
\begin{equation}
K_j=\frac{1}{2}\sum\limits_{l=1}^2B_{lj}+\pi\ii\left(j-2\right).
\end{equation} 
Next, we define the inner parametrices in the neighbourhood of $\pm\ii\eta_1, \pm\ii\eta_2$. This inner parametrices should satisfy the jump condition of $\mathbf{T}(\lambda; x, 0)$ in the neighbourhood of these points, and they also will match the outer parametrix at some small distance independent on the large variable $x$. We first consider the inner parametrix in the neighbourhood of $\ii\eta_2$. Set a small disc $B_{\rho}^{\ii\eta_2}$ centered at $\ii\eta_2$ with the radius $\rho$. Define a conformal map as:
\begin{equation}
\zeta=\frac{1}{4}\left\{x\left[-g(\lambda)+\ii\lambda\right]\right\}^2.
\end{equation}
Then we introduce an auxiliary matrix function $\mathbf{U}_{\rm Bes}(\zeta)$,
\begin{equation}
\mathbf{U}_{\rm Bes}(\zeta):=\mathbf{T}(\lambda; x, 0)\left(\frac{\ii}{\sqrt{2\ii r_1}f}\right)^{\sigma_3}\ee^{-2\zeta^{1/2}\sigma_3}\begin{bmatrix}0&1\\1&0
\end{bmatrix},\quad\lambda\in B_{\rho}^{\ii\eta_2}\cap \mathbb{C},
\end{equation}
and here we choose the positive square root, that is $\zeta^{1/2}=\frac{1}{2}\left[x\left(-g(\lambda)+\ii\lambda\right)\right]$. Then the jumps of $\mathbf{U}_{\rm Bes}(\zeta)$ with the variable $\zeta$ are locally taken into three rays $\arg(\zeta)=\pm\frac{2}{3}\pi, \arg(\zeta)=\pi$(see Fig.\ref{fig:bess}),
\begin{figure}[ht]
	\centering
	\includegraphics[width=0.3\textwidth]{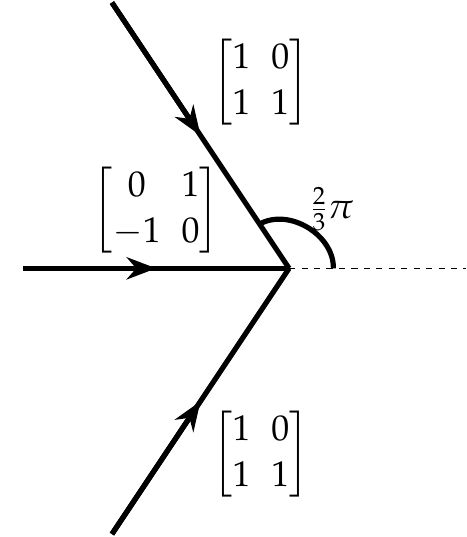}
	\caption{The jump conditions of $\mathbf{U}_{\rm Bes}(\zeta)$ in the neighbourhood of $\ii\eta_2$.}
	\label{fig:bess}
\end{figure}
\begin{equation}
\mathbf{U}_{\rm Bes, +}(\zeta)=\mathbf{U}_{\rm Bes, -}(\zeta)\left\{\begin{aligned}&\begin{bmatrix}1&0\\
1&1
\end{bmatrix},\quad\,\,\,\, \arg(\zeta)=\pm\frac{2}{3}\pi,\\
&\begin{bmatrix}0&1\\
-1&0
\end{bmatrix},\quad\arg(\zeta)=\pi.
\end{aligned}\right.
\end{equation}
We now proceed to define a matrix function $\mathbf{U}_{\rm Bes}(\zeta)$ satisfying the following Riemann-Hilbert problem.
\begin{rhp}(Riemann-Hilbert problem of the modified Bessel function)
\begin{itemize}
\item {\bf Analyticity:} $\mathbf{U}_{\rm Bes}(\zeta)$ is analytic in the $\zeta$-plane except for $\arg(\zeta)=\pm\frac{2}{3}\pi, \arg(\zeta)=\pi$(see Fig.\ref{fig:bess}).
\item {\bf Jump condition:} $\mathbf{U}_{\rm Bes, +}(\zeta)=\mathbf{U}_{\rm Bes, -}(\zeta)\mathbf{V}(\zeta)$, where $\mathbf{V}(\zeta)$ is a matrix given in the jump contours, which is shown in Fig.\ref{fig:bess}.
\item{\bf Normalization:} As $\zeta\to 0$, $\mathbf{U}_{\rm Bes}(\zeta)=\begin{bmatrix}\mathcal{O}\left(\log |\zeta|\right)&\mathcal{O}\left(\log|\zeta|\right)\\
\mathcal{O}\left(\log|\zeta|\right)&\mathcal{O}\left(\log|\zeta|\right)
\end{bmatrix}$.
\end{itemize}
\end{rhp}
The solution of this Riemann-Hilbert problem can be expressed with the modified Bessel function, and we do dot give any detailed calculation for it. In particular, when $\zeta\to\infty$, $\mathbf{U}_{\rm Bes}(\zeta)$ exhibits the following asymptotic property:
\begin{equation}
\mathbf{U}_{\rm Bes}(\zeta)=\left(2\pi\right)^{-\frac{1}{2}\sigma_3}\zeta^{-\sigma_3/4}\frac{1}{\sqrt{2}}\begin{bmatrix}1&\ii\\\ii&1
\end{bmatrix}\left(\mathbb{I}+\mathcal{O}\left(\frac{1}{\zeta^{1/2}}\right)\right)\ee^{2\zeta^{1/2}\sigma_3}.
\end{equation}
From $\mathbf{U}_{\rm Bes}(\zeta)$, we construct the inner parametrix near $\lambda=\ii\eta_2$. Define
\begin{equation}
\dot{\mathbf{T}}^{\ii\eta_2}(\lambda; x)=\mathbf{A}(\lambda)\mathbf{U}_{\rm Bes}(\zeta)\begin{bmatrix}0&1\\
1&0
\end{bmatrix}\ee^{2\zeta^{1/2}\sigma_3}\left(\frac{\ii}{\sqrt{2\ii r_1}f}\right)^{-\sigma_3},
\end{equation}
where $\mathbf{A}(\lambda)$ is a prefactor that should match the outer parametrix by the following constraint condition,
\begin{equation}\label{eq:cons-cond}
\dot{\mathbf{T}}^{\ii\eta_2}(\lambda; x)\left(\dot{\mathbf{T}}^{\rm out}(\lambda; x, 0)\right)^{-1}=\mathbb{I}+\mathcal{O}\left(|x|^{-1}\right),\,\, \text{as}\,\, x\to-\infty, \,\,\text{for}\,\,\lambda\in\partial B_{\rho}^{\ii\eta_2}.
\end{equation} 
Then this prefactor $\mathbf{A}(\lambda)$ can be given uniquely
\begin{equation}
\mathbf{A}(\lambda)=\dot{\mathbf{T}}^{\rm out}(\lambda; x, 0)\left(\frac{\ii}{\sqrt{2\ii r_1}f}\right)^{\sigma_3}\frac{1}{\sqrt{2}}\begin{bmatrix}-\ii&1\\
1&-\ii
\end{bmatrix}\left(2\pi \zeta^{1/2}\right)^{\frac{1}{2}\sigma_3}. 
\end{equation}
For the other branch points, we can similarly define the inner parametrices, and we do not give the details any more. Then the global parametrix of $\mathbf{T}(\lambda; x, 0)$ is then defined as:
\begin{equation}
\dot{\mathbf{T}}(\lambda; x, 0)=\left\{\begin{aligned}&\dot{\mathbf{T}}^{\ii\eta_1}(\lambda; x),&&\quad \lambda\in B_{\rho}^{\ii\eta_1},\\
&\dot{\mathbf{T}}^{\ii\eta_2}(\lambda; x),&&\quad \lambda\in B_{\rho}^{\ii\eta_2},\\
&\dot{\mathbf{T}}^{-\ii\eta_1}(\lambda; x),&&\quad \lambda\in B_{\rho}^{-\ii\eta_1},\\
&\dot{\mathbf{T}}^{-\ii\eta_2}(\lambda; x),&&\quad \lambda\in B_{\rho}^{-\ii\eta_2},\\
&\dot{\mathbf{T}}^{\rm out}(\lambda; x),&&\quad \lambda\in\mathbb{C}\setminus\left(\Sigma_{\pm}\cup\Sigma_{c}\cup\Gamma_{\pm}
\cup\overline{B_{\rho}^{\ii\eta_1}}\cup\overline{B_{\rho}^{-\ii\eta_1}}\cup\overline{B_{\rho}^{\ii\eta_2}}\cup\overline{B_{\rho}^{-\ii\eta_2}}\right).
\end{aligned}\right.
\end{equation}
{\emph Error Analysis:} It is straightforward to see that the error matrix $\mathbf{F}(\lambda; x):=\mathbf{T}(\lambda; x, 0)\left(\dot{\mathbf{T}}(\lambda; x, 0)\right)^{-1}$ satisfies the conditions of small-norm Riemann-Hilbert problem. The jump contours $\Sigma_{\mathbf{F}(\lambda; x)}$ of this error matrix $\mathbf{F}(\lambda; x)$ is composed with the arcs $C_{L}^{\pm}\cup C_{R}^{\pm}\cup \partial B_{\rho}^{\pm\ii\eta_1}\cup\partial B_{\rho}^{\pm\ii\eta_2}$. From our imposing condition Eq.\eqref{eq:cons-cond}, we know that the inner parametrices have the following key properties,
\begin{equation}
\begin{aligned}
&\|\dot{\mathbf{T}}^{\pm\ii\eta_1, \pm\ii\eta_2}(\lambda; x)\left(\dot{\mathbf{T}}^{\rm out}(\lambda; x, 0)\right)^{-1}-\mathbb{I}\|=\mathcal{O}(|x|^{-1}),&\quad x\to-\infty,\\
&\|\mathbf{T}(\lambda; x, 0)\left(\dot{\mathbf{T}}^{\rm out}(\lambda; x, 0)\right)^{-1}-\mathbb{I}\|=\ee^{\mu x}(\mu>0),&\quad x\to-\infty.
\end{aligned}
\end{equation}
With the small-norm Riemann-Hilbert problem, we find that $\mathbf{F}(\lambda; x)-\mathbb{I}=\mathcal{O}(|x|^{-1})$ as $x\to-\infty$. Taking the Laurent series expansion of $\mathbf{F}(\lambda; x)$ as $\lambda\to\infty$, 
\begin{equation}
\mathbf{F}(\lambda; x)=\mathbb{I}+\sum\limits_{m=1}^{\infty}\lambda^{-m}\mathbf{F}^{m}(x),
\end{equation}
where the coefficients $\mathbf{F}^{m}(x)$ can be expressed into an integral form via the matrix Plemelj formula,
\begin{equation}
\mathbf{F}^{m}(x)=-\frac{1}{2\pi\ii}\int_{\Sigma_{\mathbf{F}}}\mathbf{F}_{-}(\xi; x)\left(\mathbf{V}^{\mathbf{F}}(\xi; x)-\mathbb{I}\right)\xi^{m-1}d\xi.
\end{equation} 
Therefore, the coefficients $\mathbf{F}^{m}(x)$ satisfy the condition $\|\mathbf{F}^{m}(x)\|=\mathcal{O}\left(|x|^{-1}\right)$ as $x\to-\infty$. 
Then the potential $q(x; 0)$ at large negative $x$ is given by as follows: 
\begin{equation}
\begin{aligned}
q(x,0)&=2\ii\lim\limits_{\lambda\to\infty}\lambda \mathbf{Q}(\lambda; x, 0)_{12}=2\ii\lim\limits_{\lambda\to\infty}\mathbf{T}(\lambda; x, 0)_{12}\ee^{-xg(\infty)}f_{\infty}^2\\
&=2\ii\lim\limits_{\lambda\to\infty}\lambda\left(\mathbf{F}(\lambda; x)_{11}\dot{\mathbf{T}}^{\rm out}(\lambda; x, 0)_{12}+\mathbf{F}(\lambda; x)_{12}\dot{\mathbf{T}}^{\rm out}(\lambda; x, 0)_{22}\right)\ee^{-xg(\infty)}f_{\infty}^2\\
&=2\ii\lim\limits_{\lambda\to\infty}\lambda\dot{\mathbf{T}}^{\rm out}(\lambda; x, 0)_{12}f_{\infty}^2+\mathcal{O}(|x|^{-1}),
\end{aligned}
\end{equation} 
substituting Eq.\eqref{eq:O-matrix} and Eq.\eqref{eq:O-matrix-1} into $\dot{\mathbf{T}}^{\rm out}(\lambda; x, 0)_{12}$ yields the final asymptotic expression of $q(x,0)$ as $x\to-\infty$, which is shown in the following theorem.
\begin{theorem}\label{theo:largex}(Large $x$ asymptotics as $x\to-\infty$) For $t=0$, as $x\to-\infty$, the large $x$ asymptotics of the soliton gas for the mKdV equation can be expressed with the Riemann-Theta function, 
\begin{multline}\label{eq:q0}
q(x,0)=\frac{\Theta\left(\mathbf{A}(\infty)+\mathbf{d}\right)}{\Theta\left(\mathbf{A}(\infty)+\mathbf{d}-\pmb{\mathcal{U}}F_{1}-\pmb{\mathcal{V}}F_{2}\right)}	\frac{\Theta\left(\mathbf{A}(\infty)-\mathbf{d}+\pmb{\mathcal{U}}F_{1}+\pmb{\mathcal{V}}F_{2}\right)}{\Theta\left(\mathbf{A}(\infty)-\mathbf{d}\right)}\\
	\times \ii f_{\infty} ^2\left(\eta_1-\eta_2-1\right)\ee^{2F_{1}J_{1}+2F_{2}J_{2}-2F_{0}}+\mathcal{O}(|x|^{-1}),
\end{multline} 
where $\mathbf{A}, \mathbf{d}, \pmb{\mathcal{V}}, \pmb{\mathcal{U}}, F_0, F_1, F_2, J_1, J_2, f_{\infty} $ are defined in Eq. \eqref{eq:Aj}, Eq. \eqref{eq:d}, Eq. \eqref{eq:UjVj}, Eq. \eqref{eq:F2F1F0}, Eq. \eqref{eq:J1J2}, Eq.\eqref{eq:finfinity} respectively. $\eta_1, \eta_2$ are determined by the discrete spectra, defined in Eq. \eqref{eq:spectral}.  
\end{theorem}
\section{The asymptotics of the potential $q(x,t)$ as $t\to\infty$}
\label{sec:larget}
In the last section, we studied the asymptotics as $x\to\pm\infty$. When $x\to+\infty$, this potential decays the constant background exponentially. While when $x\to-\infty$, the leading-order term of this potential can be expressed with a Riemann-Theta function attached to a Riemann surface with genus-two. In this section, we will study the asymptotics as $t\to\infty$. Observing the jump conditions in Eq.\eqref{eq:jump-P-1}, when $\xi:=\frac{x}{t}>2+4\eta_2^2$, the first two jump matrices in Eq.\eqref{eq:jump-P-1} will decay to the identity matrix exponentially. Then the large $t$ asymptotics is similar to the case of large $x$ asymptotics when $x\to+\infty$. In this case, the large $t$ asymptotics can be given as 
\begin{equation}\label{eq:qt2}
	q(x,t)=1+\mathcal{O}\left(\ee^{-t\mu_1}\right), (\mu_1>0).
	\end{equation}
In the more interesting case where $\xi<2+4\eta_2^2$, we encounter two subcases, one is $\xi_{\rm crit}<\xi<2+4\eta_2^2$ and the other one is $\xi<\xi_{\rm crit}$, where $\xi_{\rm crit}$ is determined in the following calculation \eqref{eq:xi-crit}. The distinction between these subcases lies in the dependence of one result on a variable branch point parameter $\alpha$, while the other one is independent on this parameter.  Next, we give a detailed description to these two subcases. 
\subsection{Large $t$ asymptotics depending on the parameter $\alpha$}
In this subsection, we focus on the case where $\xi_{\rm crit}<\xi<2+4\eta_2^2$. For simplicity, we define the intervals as $[\ii\alpha, \ii\eta_2]=\Sigma_{\alpha,+}, \left[\ii,\ii\alpha\right]=\Gamma_{\alpha,+}, [-\ii, -\ii\alpha]=\Gamma_{\alpha,-}, [-\ii\alpha, -\ii\eta_2]=\Sigma_{\alpha,-}$, where $\alpha$ is a function of $\xi$ satisfying the condition $\alpha\in\left(\eta_1, \eta_2\right)$. In the last section, we have constructed a new Riemann-Hilbert problem for $\mathbf{Q}(\lambda; x, t)$ to eliminate the factor $\rho(\lambda)$ in the phase term. Now, we will explore the large $t$ asymptotics using this modified Riemann-Hilbert problem. By introducing the sign $\xi$, we rewrite the jump matrices of $\mathbf{Q}(\lambda; x, t)$ as follows:
	\begin{equation}\label{eq:jump-Q-1}
	\begin{aligned}
		&\mathbf{Q}_{+}(\lambda; x, t)=\mathbf{Q}_{-}(\lambda; x, t)\ee^{-\ii t(\lambda\xi+4\lambda^3)\sigma_3}\begin{bmatrix}1&0\\-2\ii r_1(\lambda)&1
		\end{bmatrix}\ee^{\ii t(\lambda \xi+4\lambda^3)\sigma_3},\quad &&\lambda\in\left[\ii\eta_1, \ii\eta_2\right],\\
		&\mathbf{Q}_{+}(\lambda; x, t)=\mathbf{Q}_{-}(\lambda; x, t)\ee^{-\ii t(\lambda \xi+4\lambda^3)\sigma_3}\begin{bmatrix}1&-2\ii r_1(\lambda)\\0&1
		\end{bmatrix}\ee^{\ii t(\lambda\xi+4\lambda^3 )\sigma_3},\quad &&\lambda\in\left[-\ii\eta_1,-\ii\eta_2\right],\\
		&\mathbf{Q}_{+}(\lambda; x, t)=\mathbf{Q}_{-}(\lambda; x, t)\ee^{-\ii t(\lambda\xi+4\lambda^3 )\sigma_3}\begin{bmatrix}0&\ii\\\ii&0
		\end{bmatrix}\ee^{\ii t(\lambda\xi+4\lambda^3 )\sigma_3},\quad&&\lambda\in\Sigma_c.
	\end{aligned}
\end{equation} 
In order to investigate the large $t$ asymptotics in the regime where $\xi_{\rm crit}<\xi<2+4\eta_2^2$, we introduce a $G$-function that satisfies the following Riemann-Hilbert problem.
\begin{rhp}($G$-function in the large $t$ asymptotics)
	\begin{itemize}
		\item {\bf Analyticity:} $G$-function is analytic for $\lambda\in\mathbb{C}\setminus\left(\Sigma_{\alpha, \pm}\cup\Gamma_{\alpha, \pm}\cup\Sigma_c\right)$.
		\item {\bf Jump condition:} $G$-function takes continuous values on these intervals, which are related by the following jump conditions,
		\begin{equation}\label{eq:G-jump}
			\begin{aligned}
				&G_{+}(\lambda)+G_{-}(\lambda)=2\ii\lambda\xi+8\ii\lambda^3+\kappa_{\alpha,1},\quad&&\lambda\in\Sigma_{\alpha,+},\\
				&G_{+}(\lambda)-G_{-}(\lambda)=d_{\alpha, 1}, \quad&&\lambda\in\Gamma_{\alpha, +}\cup\Gamma_{\alpha,+}^{[1]},\\
				&G_{+}(\lambda)+G_{-}(\lambda)=2\ii\lambda\xi+8\ii\lambda^3+\kappa_{\alpha, 2},\quad&&\lambda\in\Sigma_c,\\
				&G_{+}(\lambda)-G_{-}(\lambda)=d_{\alpha, 2}, \quad&&\lambda\in\Gamma_{\alpha, -}\cup\Gamma_{\alpha, -}^{[1]},\\
				&G_{+}(\lambda)+G_{-}(\lambda)=2\ii\lambda\xi+8\ii\lambda^3+\kappa_{\alpha, 3},\quad&&\lambda\in\Sigma_{\alpha, -}.
			\end{aligned}
		\end{equation}
		\item {\bf Normalization:} As $\lambda\to\infty$, $G(\lambda)$ satisfies the normalization condition, 
		\begin{equation}
			G(\lambda)\to\mathcal{O}(\lambda^{-1}).
		\end{equation}
		\item {\bf Symmetry:} $G(\lambda)$ has the symmetry condition: $G(\lambda)=-G(\lambda^*)^*.$
	\end{itemize}
\end{rhp}  

To solve this Riemann-Hilbert problem, we differentiate this $G$-function to eliminate the integral constants $\kappa_{\alpha, 1}, \kappa_{\alpha, 2}, \kappa_{\alpha, 3}, d_{\alpha, 1}, d_{\alpha, 2}$. That is
\begin{equation}
	\begin{split}
	G'_{+}(\lambda)+G'_{-}(\lambda)=2\ii\xi+24\ii\lambda^2,\quad\lambda\in\Sigma_{\alpha,+}\cup\Sigma_{c}\cup\Sigma_{\alpha, -}.
	\end{split}
	\end{equation}
To solve this problem, we introduce a $R_{\alpha}(\lambda)$ function with the following definition: 
\begin{equation}
	R_{\alpha}(\lambda):=\sqrt{(\lambda-\ii)(\lambda+\ii)(\lambda-\ii\alpha)(\lambda+\ii\alpha)(\lambda+\ii\eta_2)(\lambda-\ii\eta_2)}.
	\end{equation}
Following the calculation steps of $g'(\lambda)$ in the last section, we can set $G'(\lambda)$ as follows:
\begin{equation}
	G'(\lambda)=\ii\xi-\frac{\ii\lambda^3+\ii c_{\alpha,1}\lambda}{R_\alpha(\lambda)}\xi+12\ii\lambda^2-\frac{12\ii\lambda^5+6\ii\left(1+\alpha^2+\eta_2^2\right)\lambda^3+\ii c_{\alpha,3}\lambda}{R_{\alpha}(\lambda)},
	\end{equation}
where the unknown parameters $c_{\alpha, 1}, c_{\alpha, 3}$ satisfy the constraint condition,
\begin{equation}
	c_{\alpha, 1}=-\frac{\int_{\ii\alpha}^{\ii}\frac{\lambda^3}{R_{\alpha}(\lambda)}d\lambda}{\int_{\ii\alpha}^{\ii}\frac{\lambda}{R_{\alpha}(\lambda)}d\lambda}, \quad c_{\alpha, 3}=-\frac{\int_{\ii\alpha}^{\ii}\frac{12\lambda^5+6\left(1+\alpha^2+\eta_2^2\right)\lambda^3}{R_{\alpha}(\lambda)}d\lambda}{\int_{\ii\alpha}^{\ii}\frac{\lambda}{R_{\alpha}(\lambda)}d\lambda},
 	\end{equation}
 to ensure that the integral constants $\kappa_{\alpha, 2}, \kappa_{\alpha, 3}$ are pure imaginary. 
 Moreover, we can simplify the parameters $c_{\alpha, 1}, c_{\alpha, 3}$ to other formulas,
 \begin{equation}\label{eq:c1c3}
 c_{\alpha,1}=\eta_2^2-\frac{(\eta_2^2-1)E\left(m_{\alpha}\right)}{K\left(m_{\alpha}\right)}, \quad c_{\alpha, 3}=4(\eta_2^2+\alpha^2(\eta_2^2+1))-2(1+\alpha^2+\eta_2^2)c_{\alpha, 1},
 \end{equation}
 where $m_{\alpha}=\sqrt{\frac{\alpha^2-1}{\eta_2^2-1}}, K\left(m_{\alpha}\right), E\left(m_{\alpha}\right)$ are the complete elliptic integrals of the first and second kind respectively. 
 And $G(\lambda)$ can be written as an integral formula:
 \begin{equation}\label{eq:G-function}
 G(\lambda)=\int_{\infty}^{\lambda}\left(\ii\xi-\frac{\ii\zeta^3+\ii c_{\alpha,1}\zeta}{R_\alpha(\zeta)}\xi+12\ii\zeta^2-\frac{12\ii\zeta^5+6\ii\left(1+\alpha^2+\eta_2^2\right)\zeta^3+\ii c_{\alpha,3}\zeta}{R_{\alpha}(\zeta)}\right)d\zeta.
 \end{equation}
 Near the branch points $\pm\ii\eta_2, \pm\ii\alpha, \pm\ii$, we impose the function $G'(\lambda)-\ii\xi-12\ii\lambda^2$ subject to the constraint condition,
 \begin{equation}\label{eq:constr-cond}
 	\begin{aligned}
 	&G'(\lambda)-\ii\xi-12\ii\lambda^2=\mathcal{O}\left(\lambda\pm\ii\eta_2\right)^{-1/2}, &&\quad \lambda\to\mp\ii\eta_2,\\
 	&G'(\lambda)-\ii\xi-12\ii\lambda^2=\mathcal{O}\left(\lambda\pm\ii\right)^{-1/2}, &&\quad \lambda\to\mp\ii,\\
 	&G'(\lambda)-\ii\xi-12\ii\lambda^2=\mathcal{O}\left(\lambda\pm\ii\alpha\right)^{1/2},&&\quad \lambda\to\mp\ii\alpha.
 	\end{aligned}
 	\end{equation}
 Then we can establish a relationship between the parameter $\xi$ and the branch points $\pm\ii\alpha$, 
 \begin{equation}\label{eq:xi}
 \begin{aligned}
 	\xi&=-\frac{12\alpha^4-6\left(1+\alpha^2+\eta_2^2\right)\alpha^2+c_{\alpha, 3}}{-\alpha^2+c_{\alpha, 1}}=2(1+\alpha^2+\eta_2^2)+4\frac{(\alpha^2-1)(\alpha^2-\eta_2^2)}{\alpha^2-\eta_2^2+\frac{(\eta_2^2-1)E(m_\alpha)}{K(m_\alpha)}}\\
 &=6+2\left(\eta_2^2-1\right)\left[m_{\alpha}^2+1+2\frac{m_{\alpha}^2\left(m_{\alpha}^2-1\right)}{m_{\alpha}^2-1+\frac{E(m_\alpha)}{K(m_\alpha)}}\right].
 \end{aligned}
 	\end{equation}
 Similar to the calculation in \cite{girotti2021rigorous}, we know that $\xi$ is monotone increasing function with respect to the parameter $\alpha$ when $1<\alpha<\eta_2$. Thus we have the following inequality, 
 \begin{equation}\label{eq:xiinequality}
 2(2+\eta_2^2)<\xi<2(1+2\eta_2^2). 
 \end{equation}
 When $\ii\alpha$ moves to the point $\ii\eta_1$, we get a special value of $\xi$, which we denote as $\xi_{\rm crit}$,
 \begin{equation}\label{eq:xi-crit}
 	\xi_{\rm crit}=-\frac{12\eta_1^4-6\left(1+\eta_1^2+\eta_2^2\right)\eta_1^2+c_{\alpha, 3}}{-\eta_1^2+c_{\alpha, 1}}.
 	\end{equation}
 Similar to the Proposition \ref{prop:inter-cons} in the large negative $x$ asymptotics, the integral constants $\kappa_{\alpha, 1}, \kappa_{\alpha, 2}, \kappa_{\alpha, 3}$ in the branch cut are all zero. And these integral constants $d_{\alpha, 1}, d_{\alpha, 2}$ can be given as: 
 \begin{equation}
 \begin{split}
 	d_{\alpha, 1}&=-2\int_{\ii\eta_2}^{\ii\alpha}\frac{\ii\lambda^3+\ii c_{\alpha, 1}\lambda}{R_{\alpha}(\lambda)}\xi-2\int_{\ii\eta_2}^{\ii\alpha}\frac{12\ii\lambda^5+6\ii\left(1+\alpha^2+\eta_2^2\right)\lambda^3+\ii c_{\alpha, 3}\lambda}{R_{\alpha}(\lambda)}d\lambda,\\
 	d_{\alpha, 2}&=-2\int_{-\ii\eta_2}^{-\ii\alpha}\frac{\ii\lambda^3+\ii c_{\alpha, 1}\lambda}{R_{\alpha}(\lambda)}\xi-2\int_{-\ii\eta_2}^{-\ii\alpha}\frac{12\ii\lambda^5+6\ii\left(1+\alpha^2+\eta_2^2\right)\lambda^3+\ii c_{\alpha, 3}\lambda}{R_{\alpha}(\lambda)}d\lambda.
 	\end{split}
 	\end{equation}
 With this newly defined $G$-function, we will analyze the large $t$ asymptotics using the steepest-descent method.  Define 
 \begin{equation}\label{eq:Txt}
 	\mathbf{T}(\lambda; x, t):=\left\{\begin{aligned}
 		&f_{\infty}^{-\sigma_3}\mathbf{Q}(\lambda; x, t)\ee^{-\ii t\left(\lambda\xi+4\lambda^3\right)\sigma_3}\begin{bmatrix}1&\frac{1}{-2\ii r_1(\lambda)}\\0&1
 		\end{bmatrix}\ee^{\ii t\left(\lambda\xi+4\lambda^3 \right)\sigma_3}\ee^{-tG(\lambda)\sigma_3}f(\lambda)^{\sigma_3},\quad&&\lambda\in R_{\alpha, +},\\
 		&f_{\infty}^{-\sigma_3}\mathbf{Q}(\lambda; x, t)\ee^{-\ii t\left(\lambda\xi+4\lambda^3\right)\sigma_3}\begin{bmatrix}1&\frac{1}{2\ii r_1(\lambda)}\\
 			0&1\end{bmatrix}\ee^{\ii t\left(\lambda\xi+4\lambda^3\right)\sigma_3}\ee^{-tG(\lambda)\sigma_3}f(\lambda)^{\sigma_3},\quad&&\lambda\in L_{\alpha, +},\\
 		&f_{\infty}^{-\sigma_3}\mathbf{Q}(\lambda; x, t)\ee^{-\ii t\left(\lambda\xi+4\lambda^3\right)\sigma_3}\begin{bmatrix}1&0\\\frac{1}{-2\ii r_1(\lambda)}&1
 		\end{bmatrix}\ee^{\ii t\left(\lambda\xi+4\lambda^3\right)\sigma_3}\ee^{-tG(\lambda)\sigma_3}f(\lambda)^{\sigma_3},\quad&&\lambda\in R_{\alpha, -},\\
 		&f_{\infty}^{-\sigma_3}\mathbf{Q}(\lambda; x, t)\ee^{-\ii t\left(\lambda\xi+4\lambda^3\right)\sigma_3}\begin{bmatrix}1&0\\\frac{1}{2\ii r_1(\lambda)}&
 			1\end{bmatrix}\ee^{\ii t\left(\lambda\xi+4\lambda^3\right)\sigma_3}\ee^{-tG(\lambda)\sigma_3}f(\lambda)^{\sigma_3},\quad&&\lambda\in L_{\alpha, -},\\
 		&f_{\infty}^{-\sigma_3}\mathbf{Q}(\lambda; x, t)\ee^{-tG(\lambda)\sigma_3}f(\lambda)^{\sigma_3},\quad&&\text{otherwise}.
 	\end{aligned}\right.
 \end{equation}
A direct calculation obtains the jump matrices for $\mathbf{T}(\lambda; x, t)$,
\begin{equation}\label{eq:jump-T-1}
	\begin{aligned}
		&\mathbf{T}_{+}(\lambda; x, t)=\mathbf{T}_{-}(\lambda; x, t)\begin{bmatrix}
			1&-\frac{1}{2\ii r_1(\lambda)}\ee^{-2\ii t\left(\lambda\xi+4\lambda^3\right)+2tG(\lambda)}\frac{1}{f^2}\\
			0&1
		\end{bmatrix},&&\quad\lambda\in C_{\alpha, L}^{+},\\
		&\mathbf{T}_{+}(\lambda; x, t)=\mathbf{T}_{-}(\lambda; x, t)\begin{bmatrix}0&1\\
			-1&0
		\end{bmatrix}, &&\quad\lambda\in\Sigma_{\alpha, +},\\
		&\mathbf{T}_{+}(\lambda; x, t)=\mathbf{T}_{-}(\lambda; x, t)\begin{bmatrix}1&-\frac{1}{2\ii r_1(\lambda)}\ee^{-2\ii t\left(\lambda\xi+4\lambda^3\right)+2tG(\lambda)}\frac{1}{f^2}\\
			0&1
		\end{bmatrix},&&\quad\lambda\in C_{\alpha, R}^{+},\\
	&\mathbf{T}_{+}(\lambda; x, t)=\mathbf{T}_{-}(\lambda; x, t)\begin{bmatrix}\ee^{-td_{\alpha, 1}+\Delta_{1}}&0\\
		-2\ii r_{1}(\lambda)f_{+}(\lambda)f_{-}(\lambda)\ee^{2\ii t\left(\lambda\xi+4\lambda^3\right)-t\left(G_{+}(\lambda)+G_{-}(\lambda)\right)}&\ee^{td_{\alpha, 1}+\Delta_{1}}
		\end{bmatrix},&&\quad\lambda\in\Gamma_{\alpha,+}^{[1]},\\
		&\mathbf{T}_{+}(\lambda; x, t)=\mathbf{T}_{-}(\lambda; x, t)\begin{bmatrix}\ee^{-td_{\alpha, 1}+\Delta_1}&0\\0&\ee^{td_{\alpha, 1}-\Delta_1}
		\end{bmatrix},&&\quad\lambda\in\Gamma_{\alpha, +},\\
		&\mathbf{T}_{+}(\lambda; x, t)=\mathbf{T}_{-}(\lambda; x, t)\begin{bmatrix}
			0&1\\
			-1&0	\end{bmatrix},&&\quad\lambda\in\Sigma_c,\\
		&\mathbf{T}_{+}(\lambda; x, t)=\mathbf{T}_{-}(\lambda; x, t)\begin{bmatrix}\ee^{-td_{\alpha, 1}+\Delta_2}&0\\
			0&\ee^{td_{\alpha, 1}-\Delta_2}
		\end{bmatrix},&&\quad\lambda\in\Gamma_{\alpha, -},\\
	&\mathbf{T}_{+}(\lambda; x, t)=\mathbf{T}_{-}(\lambda; x, t)\begin{bmatrix}\ee^{-td_{\alpha, 1}+\Delta_{2}}&	\frac{-2\ii r_{1}(\lambda)}{f_{+}(\lambda)f_{-}(\lambda)}\ee^{-2\ii t\left(\lambda\xi+4\lambda^3\right)+t\left(G_{+}(\lambda)+G_{-}(\lambda)\right)}\\
	0&\ee^{td_{\alpha, 1}-\Delta_{2}}
	\end{bmatrix},&&\quad\lambda\in\Gamma_{\alpha,-}^{[1]},\\
		&\mathbf{T}_{+}(\lambda; x, t)=\mathbf{T}_{-}(\lambda; x, t)\begin{bmatrix}
			1&0\\\frac{1}{-2\ii r_1(\lambda)}\ee^{2\ii t\left(\lambda \xi+4\lambda^3\right)-2tG(\lambda)}f^2&1
		\end{bmatrix},&&\quad\lambda\in C_{\alpha, L}^{-},\\
		&\mathbf{T}_{+}(\lambda; x, t)=\mathbf{T}_{-}(\lambda; x, t)\begin{bmatrix}
			1&0\\\frac{1}{-2\ii r_1(\lambda)}\ee^{2\ii t\left(\lambda \xi+4\lambda^3\right)-2tG(\lambda)}f^2&1
		\end{bmatrix},&&\quad\lambda\in C_{\alpha, R}^{-},\\
		&\mathbf{T}_{+}(\lambda; x, t)=\mathbf{T}_{-}(\lambda; x, t)\begin{bmatrix}0&1\\
			-1&0
		\end{bmatrix}, &&\quad\lambda\in\Sigma_{\alpha, -}.
	\end{aligned}
\end{equation}
The jump contours of $\mathbf{T}(\lambda; x, t)$ are depicted in Fig.\ref{fig:jump-T-1}.
\begin{figure}[ht]
	\centering
	\includegraphics[width=0.5\textwidth]{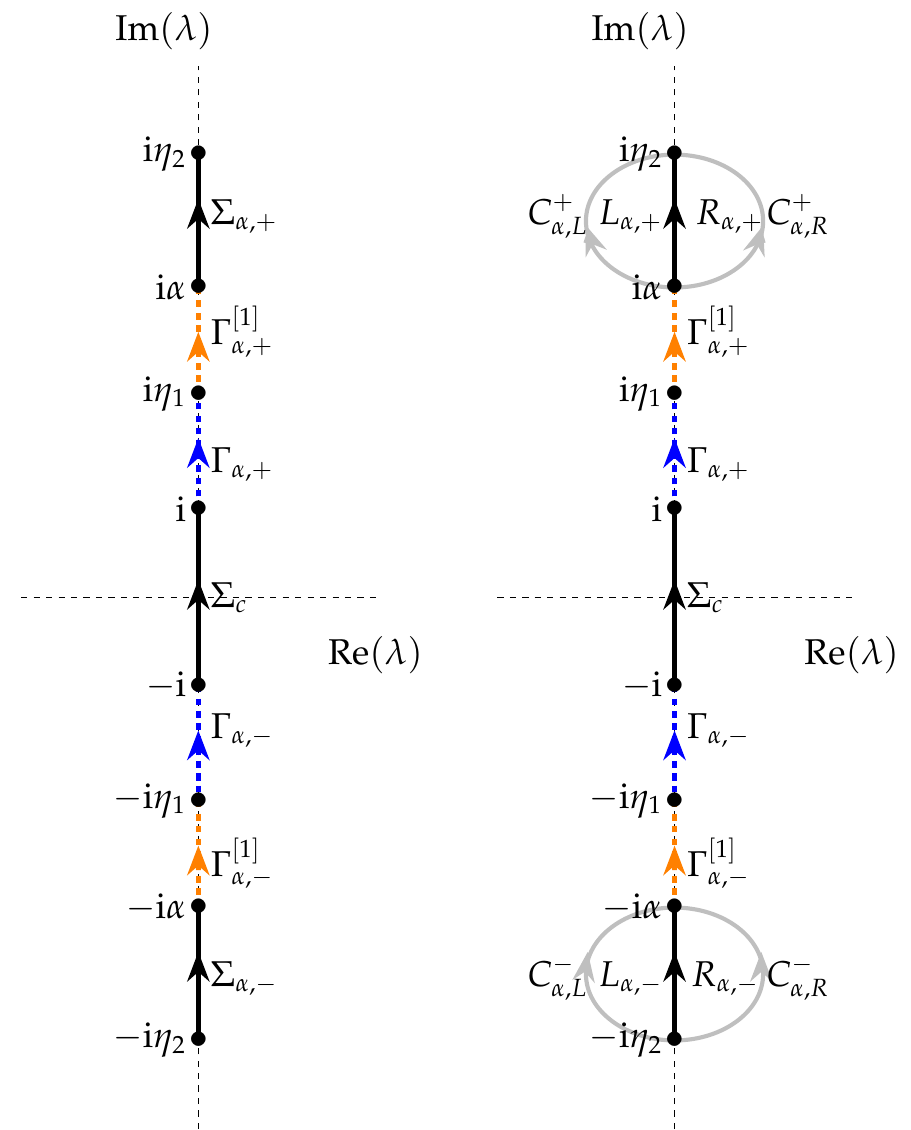}
	\caption{The left panel is the jump contour of $G$-function and the right one is the jump contour of $\mathbf{T}(\lambda; x, t)$. }
	\label{fig:jump-T-1}
\end{figure}

Next, we give a Lemma exhibiting the properties of the jump matrices of $\mathbf{T}(\lambda; x, t)$. 
\begin{lemma}\label{lemma:3}
	\begin{equation}
		\begin{aligned}
	&\Im\left(\lambda\xi+4\lambda^3+\ii G(\lambda)\right)<0,&&\quad\lambda\in C_{\alpha,L}^{+},\\
		&\Im\left(\lambda\xi+4\lambda^3+\ii G(\lambda)\right)<0,&&\quad\lambda\in C_{\alpha,R}^{+},\\
		&\Im\left(2\lambda\xi+8\lambda^3+\ii \left(G_{+}(\lambda)+G_{-}(\lambda)\right)\right)>0,&&\quad\lambda\in\Gamma_{\alpha,+}^{[1]},\\
		&\Im\left(\lambda\xi+4\lambda^3+\ii G(\lambda)\right)>0,&&\quad\lambda\in C_{\alpha,L}^{-},\\
		&\Im\left(\lambda\xi+4\lambda^3+\ii G(\lambda)\right)>0,&&\quad\lambda\in C_{\alpha,R}^{-},\\
		&\Im\left(2\lambda\xi+8\lambda^3+\ii \left(G_{+}(\lambda)+G_{-}(\lambda)\right)\right)<0,&&\quad\lambda\in\Gamma_{\alpha,-}^{[1]}.
		\end{aligned}
		\end{equation}
	\end{lemma}
\begin{proof}In the previous section, we established in Lemma \ref{lemma:2} that at $t=0$, the exponential terms in the contours $C_{L}^{\pm}, C_{R}^{\pm}$ tend to zero when $x$ is large. For large $t$ case, when $\lambda\in C_{\alpha, L}^{\pm}\cup C_{\alpha, R}^{\pm}$, the proofs are similar to the case of Lemma \ref{lemma:2}. Thus for Lemma \ref{lemma:3}, we focus solely on proving the case where $\lambda\in \Gamma_{\alpha,\pm}^{[1]}$. Referring back to the definition of $G(\lambda)$ in Eq.\eqref{eq:G-function}, it becomes evident that when $\lambda\in\Gamma_{\alpha,+}^{[1]}$, we have
\begin{multline}
G_{+}(\lambda)+G_{-}(\lambda)=2\int_{\infty}^{\ii\eta_2}\left(\ii\xi-\frac{\ii\zeta^3+\ii c_{\alpha,1}\zeta}{R_\alpha(\zeta)}\xi+12\ii\zeta^2-\frac{12\ii\zeta^5+6\ii\left(1+\alpha^2+\eta_2^2\right)\zeta^3+\ii c_{\alpha,3}\zeta}{R_{\alpha}(\zeta)}\right)d\zeta\\+2\int_{\ii\eta_2}^{\ii\alpha}\left(\ii\xi+12\ii\zeta^2\right)d\zeta
+2\int_{\ii\alpha}^{\lambda}\left(\ii\xi-\frac{\ii\zeta^3+\ii c_{\alpha,1}\zeta}{R_\alpha(\zeta)}\xi+12\ii\zeta^2-\frac{12\ii\zeta^5+6\ii\left(1+\alpha^2+\eta_2^2\right)\zeta^3+\ii c_{\alpha,3}\zeta}{R_{\alpha}(\zeta)}\right)d\zeta.
\end{multline}
Using the first relationship in Eq.\eqref{eq:G-jump} with $\kappa_{\alpha, 1}=0$, we find that when $\lambda\in\Gamma_{\alpha, +}^{[1]}$, we have 
\begin{equation}
\begin{aligned}
2\lambda\xi+8\lambda^3+\ii\left(G_{+}(\lambda)+G_{-}(\lambda)\right)&=-2\ii\int_{\ii\alpha}^{\lambda}\left(\frac{\ii\zeta^3+\ii c_{\alpha,1}\zeta}{R_\alpha(\zeta)}\xi+\frac{12\ii\zeta^5+6\ii\left(1+\alpha^2+\eta_2^2\right)\zeta^3+\ii c_{\alpha,3}\zeta}{R_{\alpha}(\zeta)}\right)d\zeta\\
&=\int_{\ii\alpha}^{\lambda}\frac{2\zeta\sqrt{\zeta^2+\alpha^2}\left(12\zeta^2+6(1-\alpha^2+\eta_2^2)+\xi\right)}{\sqrt{(\zeta^2+1)(\zeta^2+\eta_2^2)}}d\zeta,
\end{aligned}
\end{equation}
where the final expression makes use of the identity between $c_{\alpha, 3}, c_{\alpha, 1}$ as given in Eq.\eqref{eq:c1c3} and Eq.\eqref{eq:xi}. 
Set $H(\alpha, \lambda)$ as
\begin{equation}
H(\alpha, \lambda)=\frac{2\lambda\sqrt{\lambda^2+\alpha^2}\left(12\lambda^2+6(1-\alpha^2+\eta_2^2)+\xi\right)}{\sqrt{(\lambda^2+1)(\lambda^2+\eta_2^2)}}.
\end{equation}
The roots of this expression can be determined as $0, \pm\ii\alpha, \pm\ii\rho$. From the expression of $\xi$ in Eq.\eqref{eq:xi}, $H(\alpha, \lambda)$ can be rewritten as 
\begin{equation}\label{eq:H}
H(\alpha, \lambda)=\frac{2\lambda\sqrt{\lambda^2+\alpha^2}\left(12\lambda^2+12+4(\eta_2^2-1)\left(2-m_{\alpha}^2+m_{\alpha}^2
\frac{m_{\alpha}^2-1}{m_{\alpha}^2-1+\frac{E(m_{\alpha})}{K(m_{\alpha})}}\right)\right)}{\sqrt{(\lambda^2+1)(\lambda^2+\eta_2^2)}}.
\end{equation}
When $m_{\alpha}<1$, we have the following series expansion 
\begin{equation}
\frac{E(m_\alpha)}{F(m_\alpha)}=1-\frac{1}{2}m_{\alpha}^2+\mathcal{O}(m_{\alpha}^4). 
\end{equation}
Substituting it to Eq.\eqref{eq:H}, we obtain an estimation for the root $\ii\rho$, $\rho\sim\sqrt{\frac{2+\alpha^2}{3}}<\alpha$. Furthermore, we have $\Im\left(H(\alpha, \ii\alpha+\ii\epsilon)\right)<0, \Im\left(H(\alpha, \ii-\ii\epsilon)\right)>0, \Im\left(H(\alpha, \ii\beta)\right)=0(1<\beta<\alpha), \Re\left(H(\alpha, \ii\beta)\right)<0(\rho<\beta<\alpha)$ and $\Re\left(H(\alpha, \ii\beta)\right)>0(1<\beta<\rho)$. Then we have the following results:
\begin{itemize}
\item When $\ii\rho<\ii\eta_1$, we can easily get the inequality $ \Im\left(2\lambda\xi+8\lambda^3+\ii\left(G_{+}(\lambda)+G_{-}(\lambda)\right)\right)>0$ for 
$\lambda\in\left(\ii\eta_1, \ii\alpha\right)$.
\item When $\ii\eta_1<\ii\rho<\ii\alpha$, and $\lambda\in\left[\ii\rho, \ii\alpha\right]$, we have $ \Im\left(2\lambda\xi+8\lambda^3+\ii\left(G_{+}(\lambda)+G_{-}(\lambda)\right)\right)>0$.
\item When $\ii\eta_1<\ii\rho<\ii\alpha$, and $\lambda\in\left(\ii\eta_1, \ii\rho\right)$, we have $\Im(\int_{\ii\alpha}^{\ii\rho}H(\zeta, \lambda)d\zeta)>0$ and $\Im(\int_{\ii\rho}^{\lambda}H(\zeta, \lambda)d\zeta)<0$. As $\lambda=\ii\eta_1$, $\Im(\int_{\ii\alpha}^{\ii\eta_1}H(\zeta, \lambda)d\zeta)$ reach its minimum. For this case, we use a numerical method to prove the inequality$\Im(\int_{\ii\alpha}^{\ii\eta_1}H(\zeta, \lambda)d\zeta)>0$. Thus when $\lambda\in(\ii\eta_1, \ii\rho)$, we have 
$\Im(\int_{\ii\alpha}^{\lambda}H(\zeta, \lambda)d\zeta)>\Im(\int_{\ii\alpha}^{\ii\eta_1}H(\zeta, \lambda)d\zeta)>0$. 
\end{itemize} 
	\end{proof}
When $t\to\infty$, the jump matrices of $\mathbf{T}(\lambda; x, t)$ will decay exponentially to the identity matrix except for $\lambda\in\left(\Sigma_{\alpha, +}\cup\Gamma_{\alpha,+}\cup\Sigma_{c}\cup\Gamma_{\alpha,-}\cup\Sigma_{\alpha,-}\right)$. Following the step used to study the large negative $x$ asymptotics when $t=0$, we can construct the outer parametrix for $\mathbf{T}(\lambda; x, t)$. 
\begin{rhp}(Riemann-Hilbert problem of the outer parametrix $\dot{\mathbf{T}}^{\rm out}(\lambda; x, t)$)
	\begin{itemize}
		\item {\bf Analyticity:} $\dot{\mathbf{T}}^{\rm out}(\lambda; x, t)$ is analytic for $\lambda\in\mathbb{C}\setminus\left(\Sigma_{\alpha, \pm}\cup\Sigma_c\cup\Gamma_{\alpha, \pm}\right)$.
		\item {\bf Jump condition:} The boundary conditions on these contours are related by the jump relation $\dot{\mathbf{T}}^{\rm out}_{+}(\lambda; x, t)=\dot{\mathbf{T}}^{\rm out}_{-}(\lambda; x, t)\mathbf{V}_{\dot{\mathbf{T}}^{\rm out}}(\lambda; x, t)$, where 
		\begin{equation}
			\mathbf{V}_{\dot{\mathbf{T}}^{\rm out}}(\lambda; x, t)=\left\{\begin{aligned}&\begin{bmatrix}0&1\\
					-1&0
				\end{bmatrix},&&\quad \lambda\in\Sigma_{\alpha, \pm}\cup\Sigma_c,\\
				&\begin{bmatrix}\ee^{-td_{\alpha, 1}+\Delta_1}&0\\
					0&\ee^{td_{\alpha, 1}-\Delta_1}
				\end{bmatrix},&&\quad\lambda\in\Gamma_{\alpha, +},\\
				&\begin{bmatrix}\ee^{-td_{\alpha, 2}+\Delta_2}&0\\
					0&\ee^{td_{\alpha, 2}-\Delta_2}
				\end{bmatrix},&&\quad\lambda\in\Gamma_{\alpha, -}.
			\end{aligned}\right.
		\end{equation}
		\item{\bf Normalization:} $\dot{\mathbf{T}}^{\rm out}(\lambda; x, t)\to\mathbb{I}$ as $\lambda\to\infty$. 
	\end{itemize}
\end{rhp}
Observing these jump matrices of the outer parametrix $\dot{\mathbf{T}}^{\rm out}(\lambda; x, t)$, we notice their resemblance to those of the outer parametrix $\dot{\mathbf{T}}^{\rm out}(\lambda; x, 0)$ when replacing $\eta_1$ to $\alpha,$ $x$ to $t$, $d_{1}$ to $d_{\alpha, 1}$, $d_{2}$ to $d_{\alpha, 2}$. Thus we can readily deduce the solution for the outer parametrix $\dot{\mathbf{T}}^{\rm out}(\lambda; x, t)$,
\begin{multline}	
\dot{\mathbf{T}}^{\rm out}(\lambda; x, t)=\frac{1}{2}{\rm diag}\left(C_{\alpha, 1}\ee^{-F_{\alpha, 0}}, C_{\alpha, 2}\ee^{F_{\alpha, 0}}\right)\\\times\begin{bmatrix}\left(\gamma_{\alpha}( \lambda)+\frac{1}{\gamma_{\alpha}(\lambda)}\right)\frac{\Theta\left(\mathbf{A}_{\alpha}(\lambda)+\mathbf{d}_{\alpha}-\pmb{\mathcal{U}}_{\alpha}F_{\alpha, 1}-\pmb{\mathcal{V}}_{\alpha}F_{\alpha, 2}\right)}{\Theta\left(\mathbf{A}_{\alpha}(\lambda)+\mathbf{d}_{\alpha}\right)}&
	\ii\left(\gamma_{\alpha}(\lambda)-\frac{1}{\gamma_{\alpha}(\lambda)}\right)\frac{\Theta\left(\mathbf{A}_{\alpha}(\lambda)-\mathbf{d}_{\alpha}+\pmb{\mathcal{U}}_{\alpha}F_{\alpha, 1}+\pmb{\mathcal{V}}_{\alpha}F_{\alpha, 2}\right)}{\Theta\left(\mathbf{A}_{\alpha}(\lambda)-\mathbf{d}_{\alpha}\right)}\\
	-\ii\left(\gamma_{\alpha}(\lambda)-\frac{1}{\gamma_{\alpha}(\lambda)}\right)\frac{\Theta\left(\mathbf{A}_{\alpha}(\lambda)-\mathbf{d}_{\alpha}-\pmb{\mathcal{U}}_{\alpha}F_{\alpha, 1}-\pmb{\mathcal{V}}_{\alpha}F_{\alpha, 2}\right)}{\Theta\left(\mathbf{A}_{\alpha}(\lambda)-\mathbf{d}_{\alpha}\right)}&
	\left(\gamma_{\alpha}(\lambda)+\frac{1}{\gamma_{\alpha}(\lambda)}\right)\frac{\Theta\left(\mathbf{A}_{\alpha}(\lambda)+\mathbf{d}_{\alpha}+\pmb{\mathcal{U}}_{\alpha}F_{\alpha, 1}+\pmb{\mathcal{V}}_{\alpha}F_{\alpha, 2}\right)}{\Theta\left(\mathbf{A}_{\alpha}(\lambda)+\mathbf{d}_{\alpha}\right)}\\
\end{bmatrix}\\\times\ee^{-\left(\Omega_{\alpha, 1}(\lambda)F_{\alpha, 1}+\Omega_{\alpha, 2}(\lambda)F_{\alpha, 2}\right)\sigma_3}\ee^{F_{\alpha}(\lambda)\sigma_3}.
	\end{multline}
All these functions with the subscript $_{\alpha}$ denote that $\eta_1$ is replaced by $\alpha,$ $x$ by $t$, $d_{1}$ by $d_{\alpha, 1}$, $d_{2}$ by $d_{\alpha, 2}$ with their corresponding functions from Eq.\eqref{eq:O-matrix-1} without the subscript. 

Similar to the large negative $x$ asymptotics discussed in the previous section, the outer parametrix $\dot{\mathbf{T}}^{\rm out}(\lambda; x, t)$ has singularities at the branch points $\pm\ii,\pm\ii\alpha, \pm\ii\eta_2$. Although these points are all branch points, they differ in essence. The points $\pm\ii\eta_2$ are stationary, determined by the initial spectral parameter. However, the branch point $\ii\alpha$ can vary with the variables $x$ and $t$, leading to different inner parametrices. While we provided construction details for stationary branch points in the last section, which are related to modified Bessel function. The inner parametrices in the neighbourhood of the branch points $\pm\ii\alpha$ are associated with Airy function. Detailed constructions of these inner parametrices can be found in our previous work\cite{ling2022large} and in reference \cite{Bilman-arxiv-2021}, therefore we omit them in this manuscript.  The global parametrix of $\mathbf{T}(\lambda; x, t)$ can be set as:
\begin{equation}
	\dot{\mathbf{T}}(\lambda; x, t)=\left\{\begin{aligned}&\dot{\mathbf{T}}^{\ii\alpha}(\lambda; x, t),&&\quad \lambda\in B_{\rho}^{\ii\alpha},\\
		&\dot{\mathbf{T}}^{\ii\eta_2}(\lambda; x, t),&&\quad \lambda\in B_{\rho}^{\ii\eta_2},\\
		&\dot{\mathbf{T}}^{-\ii\alpha}(\lambda; x, t),&&\quad \lambda\in B_{\rho}^{-\ii\alpha},\\
		&\dot{\mathbf{T}}^{-\ii\eta_2}(\lambda; x, t),&&\quad \lambda\in B_{\rho}^{-\ii\eta_2},\\
		&\dot{\mathbf{T}}^{\rm out}(\lambda; x, t),&&\quad \lambda\in\mathbb{C}\setminus\left(\Sigma_{\pm}\cup\Sigma_{c}\cup\Gamma_{\pm}
		\cup\overline{B_{\rho}^{\ii\alpha}}\cup\overline{B_{\rho}^{-\ii\alpha}}\cup\overline{B_{\rho}^{\ii\eta_2}}\cup\overline{B_{\rho}^{-\ii\eta_2}}\right).
	\end{aligned}\right.
\end{equation}
{\emph Error Analysis} After constructing the inner parametrices, we proceed to analyze the error between $\mathbf{T}(\lambda; x, t)$ and its global parametrix $\dot{\mathbf{T}}(\lambda; x, t)$. Let us define the error matrix as $\mathbf{F}(\lambda; x, t):=\mathbf{T}(\lambda; x, t)\left(\dot{\mathbf{T}}(\lambda; x, t)\right)^{-1}$, its jump contours is denoted as $\Sigma_{\mathbf{F}(\lambda; x, t)}$. By considering the jump conditions of $\mathbf{T}(\lambda; x, t)$ and $\dot{\mathbf{T}}(\lambda; x, t)$, we observe that $\Sigma_{\mathbf{F}(\lambda; x, t)}$ is composed with the arcs $C_{\alpha, L}^{\pm}\cup C_{\alpha, R}^{\pm}\cup \partial B_{\rho}^{\pm\ii\alpha}\cup\partial B_{\rho}^{\pm\ii\eta_2}$. Leveraging the results obtained in the preceding section and previous references, we obtain the following estimates:
\begin{equation}
	\begin{aligned}
		&\|\dot{\mathbf{T}}^{\pm\ii\alpha, \pm\ii\eta_2}(\lambda; x, t)\left(\dot{\mathbf{T}}^{\rm out}(\lambda; x, t)\right)^{-1}-\mathbb{I}\|=\mathcal{O}(|t|^{-1}),&\quad t\to\infty,\\
		&\|\mathbf{T}(\lambda; x, t)\left(\dot{\mathbf{T}}^{\rm out}(\lambda; x, t)\right)^{-1}-\mathbb{I}\|=\ee^{-\mu_{1} t}(\mu_1>0),&\quad t\to\infty.
	\end{aligned}
\end{equation}
By employing the formula derived for the large negative $x$ asymptotics, we can get the large $t$ asymptotics in this case. 
\begin{theorem}\label{theo:larget}(Large $t$ asymptotics when $\xi_{\rm crit}<\xi<2+4\eta_2^2$ ) When $\xi_{\rm crit}<\xi<2+4\eta_2^2, \eta_1<\alpha<\eta_2$, where $\xi_{\rm crit}$ is defined in Eq. \eqref{eq:xi-crit}, then the large $t$ asymptotics of the soliton gas for the mKdV equation is given as: 
\begin{multline}\label{eq:qt}	q(x,t)=\frac{\Theta\left(\mathbf{A}_{\alpha}(\infty)+\mathbf{d}_{\alpha}\right)}{\Theta\left(\mathbf{A}_{\alpha}(\infty)+\mathbf{d}_{\alpha}-\pmb{\mathcal{U}}_{\alpha}F_{\alpha,1}-\pmb{\mathcal{V}}_{\alpha}F_{\alpha, 2}\right)}
	\frac{\Theta\left(\mathbf{A}_{\alpha}(\infty)-\mathbf{d}_{\alpha}+\pmb{\mathcal{U}}_{\alpha}F_{\alpha,1}+\pmb{\mathcal{V}}_{\alpha}F_{\alpha, 2}\right)}{\Theta\left(\mathbf{A}_{\alpha}(\infty)-\mathbf{d}_{\alpha}\right)}\\
	\times \ii f_{\infty} ^2\left(\alpha-\eta_2-1\right)\ee^{2F_{\alpha, 1}J_{\alpha, 1}+2F_{\alpha, 2}J_{\alpha, 2}-2F_{\alpha, 0}}+\mathcal{O}(|t|^{-1}),
\end{multline} 
where the parameters $\mathbf{A}_{\alpha}, \mathbf{d}_{\alpha}, \pmb{\mathcal{V}}_{\alpha}, \pmb{\mathcal{U}}_{\alpha}, F_{\alpha, 0}, F_{\alpha, 1}, F_{\alpha, 2}, J_{\alpha, 1}, J_{\alpha, 2}$ are similar to the parameters $\mathbf{A}, \mathbf{d}, \pmb{\mathcal{V}}, \pmb{\mathcal{U}}, F_0, F_1, F_2, J_1, J_2 $ by replacing $\eta_1$ to $\alpha,$ $x$ to $t$, $d_{1}$ to $d_{\alpha, 1}$, $d_{2}$ to $d_{\alpha, 2}$, $f_{\infty}$ is defined in Eq.\eqref{eq:finfinity}. 
\end{theorem}
\subsection{The large $t$ asymptotics in the case $\xi<\xi_{\rm crit}$}
As $\ii\alpha$ approaches the point $\ii\eta_1$, a special value $\xi_{\rm crit}$ emerges. When $\xi<\xi_{\rm crit}$, slight adjustments need to be made to the branch points of the $G$-function. During the construction of the $G$-function, there is no necessity to introduce the unknown branch point $\ii\alpha$, as it solely depends on the end points $\pm\ii\eta_2, \pm\ii\eta_1, \pm\ii$ of the spectrum. In this case, $G$-function satisfies the following jump relations,
	\begin{equation}
	\begin{aligned}
		&G_{+}(\lambda)+G_{-}(\lambda)=2\ii\lambda\xi+8\ii\lambda^3,\quad&&\lambda\in\Sigma_{+},\\
		&G_{+}(\lambda)-G_{-}(\lambda)=d_{\eta_1, 1}, \quad&&\lambda\in\Gamma_{+},\\
		&G_{+}(\lambda)+G_{-}(\lambda)=2\ii\lambda\xi+8\ii\lambda^3,\quad&&\lambda\in\Sigma_c,\\
		&G_{+}(\lambda)-G_{-}(\lambda)=d_{\eta_1, 2}, \quad&&\lambda\in\Gamma_{-},\\
		&G_{+}(\lambda)+G_{-}(\lambda)=2\ii\lambda\xi+8\ii\lambda^3,\quad&&\lambda\in\Sigma_{-}, 
	\end{aligned}
\end{equation}
where the intervals $\Sigma_{\pm}, \Gamma_{\pm}, \Sigma_{c}$ are defined in the context of the last large negative $x$ asymptotics. The integral constants $d_{\eta_1, 1}, d_{\eta, 2}$ are determined by replacing $\ii\alpha$ with $\ii\eta_1$ in the $d_{\alpha, 1}, d_{\alpha, 2}$. 
Direct calculation derives the derivative of $G$-function, 
\begin{equation}
	G'(\lambda)=\ii\xi-\frac{\ii\lambda^3+\ii c_{1}\lambda}{R(\lambda)}\xi+12\ii\lambda^2-\frac{12\ii\lambda^5+6\ii\left(1+\eta_1^2+\eta_2^2\right)\lambda^3+\ii c_{\eta_1,3}\lambda}{R(\lambda)},
\end{equation}
where $c_{1}$ is defined in \eqref{eq:c1} and $R(\lambda)$ is defined in \eqref{eq:R-function}. The parameter $c_{\eta_1, 3}$ satisfies the constraint condition
\begin{equation}
c_{\eta_1, 3}=-\frac{\int_{\ii\eta_1}^{\ii}\frac{12\lambda^5+6\left(1+\eta_1^2+\eta_2^2\right)\lambda^3}{R(\lambda)}d\lambda}{\int_{\ii\eta_1}^{\ii}\frac{\lambda}{R(\lambda)}d\lambda}.
\end{equation}
From the results of large negative $x$ asymptotics and the large $t$ asymptotics depending on the parameter $\alpha$, we can express the leading-order term in this case using the Riemann-Theta function. These Riemann-Theta functions are independent of the formulas of $g$-function, and $G$-function but are related to the integral constants $d_{1}, d_{2}$ and $d_{\alpha, 1}, d_{\alpha, 2}$ respectively. Similar to the definition of $\mathbf{T}(\lambda; x, t)$ in \eqref{eq:Txt}, $\mathbf{T}(\lambda; x, t)$ in this case can be defined using the newly defined $G$-function. The outer parametrix of $\mathbf{T}(\lambda; x, t)$ can be then given by replacing $x$ with $t$, $d_{1}$ with $d_{\eta_1, 1}$, $d_{2}$ with $d_{\eta_1, 2}$ in the case of large negative $x$ for $\mathbf{T}(\lambda; x, 0)$ as discussed in the previous section. Consequently, the large $t$ asymptotics for $\xi<\xi_{\rm crit}$ can be given directly.
\begin{theorem}\label{theo:larget-1}(Large $t$ asymptotics when $\xi<\xi_{\rm crit}$) When $\xi<\xi_{\rm crit}$, we obtain the following large $t$ asymptotics, 
\begin{multline}\label{eq:qt1} q(x,t)=\frac{\Theta\left(\mathbf{A}_{t}(\infty)+\mathbf{d}_{t}\right)}{\Theta\left(\mathbf{A}_{t}(\infty)+\mathbf{d}_{t}-\pmb{\mathcal{U}}_{t}F_{t,1}-\pmb{\mathcal{V}}_{t}F_{t, 2}\right)}	\frac{\Theta\left(\mathbf{A}_{t}(\infty)-\mathbf{d}_{t}+\pmb{\mathcal{U}}_{t}F_{t,1}+\pmb{\mathcal{V}}_{t}F_{t,2}\right)}{\Theta\left(\mathbf{A}_{t}(\infty)-\mathbf{d}_{t}\right)}\\
	\times \ii f_{\infty} ^2\left(\eta_1-\eta_2-1\right)\ee^{2F_{t,1}J_{t,1}+2F_{t,2}J_{t,2}-2F_{t,0}}+\mathcal{O}(|t|^{-1}),
\end{multline} 
where the subscript $_t$ indicates that by replacing $x$ to $t$, $d_{1}, d_{2}$ are replaced by $d_{\eta_1, 1}, d_{\eta_1, 2}$ respectively, using the corresponding functions in the Theorem \ref{theo:largex}. 
\end{theorem}
Based on these three different asymptotic expressions \eqref{eq:qt2}, \eqref{eq:qt} and \eqref{eq:qt1}, the soliton gas behavior is depicted in Fig. \ref{fig:solitongas} by choosing a fixed $t$.
\begin{figure}[ht]
	\centering
	\includegraphics[width=1\textwidth]{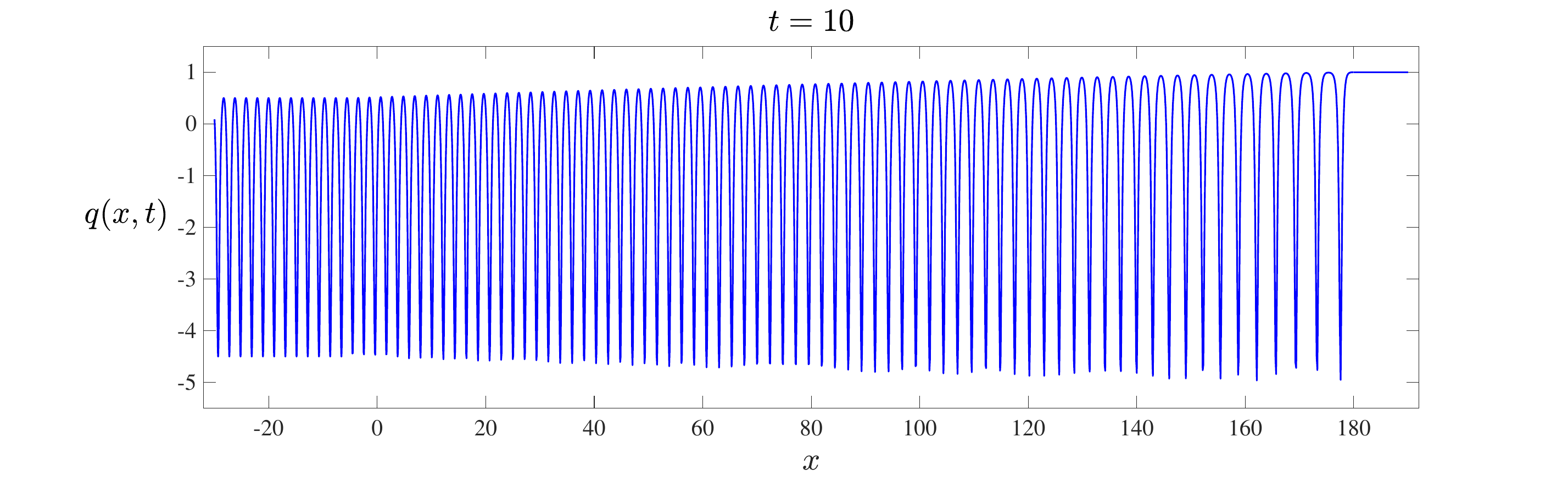}
	\caption{The large $t$ asymptotics of the solitons gas under the nonzero background, where the parameters are chosen as $t=10$, $\eta_1=\frac{3}{2}, \eta_2=2$ and $r_1(\lambda)=1$. For $x\lesssim -5.0731$, the left asymptotic region is given in Eq. \eqref{eq:qt1}, for $-5.7031\lesssim x<180$, the middle asymptotic region is given in Eq.\eqref{eq:qt}, for $x>180$, the right asymptotic region is given in Eq.\eqref{eq:qt2}.}
	\label{fig:solitongas}
\end{figure}

For the large $t$ asymptotics, the $(x, t)$-plane is divided into three different regions depending on the ratio $\xi=\frac{x}{t}$. When $\xi>2+4\eta_2^2$, the large $t$ asymptotics is given in Eq.\eqref{eq:qt2}, and the leading-order term $q(x, t)\to 1$ satisfies the mKdV equation\eqref{eq:mkdv}. For $\xi<2+4\eta_2^2$, the large $t$ asymptotics are show in Theorem \ref{theo:larget} and Theorem \ref{theo:larget-1}. By choosing fixed $(x, t)$ in the corresponding regions and substituting them to the mKdV equation \eqref{eq:mkdv}, we check the asymptotics with the numerical difference method and find that they can match the equation with the error $\mathcal{O}(t^{-1})$. 
\section{Conclusions and Discussions}
In this paper, we study the large $x$ asymptotics at $t=0$ and the large $t$ asymptotics of the soliton gas under the nonzero background for the mKdV equation. When $t=0$, the large $x$ asymptotics behave differently as $x\to-\infty$ and $x\to\infty$.  For the large $t$ asymptotics, we show that the $(x, t)$-plane is decomposed into three types of regions, $\frac{x}{t}>2+4\eta_2^2, \xi_{\rm crit}<\frac{x}{t}<2+4\eta_2^2$ and $\frac{x}{t}<\xi_{\rm crit}$, where $\xi_{\rm crit}$ is defined in Eq.\eqref{eq:xi-crit}. In the right region, the solution decays to the constant background. The asymptotics in the middle and the left regions are both expressed with the Riemann-Theta function, corresponding to a genus-two Riemann surface. 

The soliton gas in this paper is generated from the $N$-soliton solutions as $N\to\infty$, with the corresponding discrete spectra on the intervals $\left[\ii\eta_1, \ii\eta_2\right]$, where $\eta_1, \eta_2\in\mathbb{R}$ and $1<\eta_1<\eta_2$. This type of soliton gas involves the asymptotic region with genus-two, which differs from the soliton gas under the zero background. Besides the soliton solutions, the mKdV equation also admits the breather solutions under the nonzero background. These breather solutions are closely related to modulation instability and behave quite differently from soliton solutions. The ideas presented in this paper can be applied to study the asymptotics of the breather gas. Additionally, in \cite{girotti2023soliton}, Girotti et.al. studied the large $t$ asymptotics of the soliotn gas in the presence of a single trial soliton. Investigating the asymptotics of the soliton gas and a trial soliton when the soliton gas is generated under the nonzero background is a worthwhile problem. These problems will be addressed in future works. 
\subsection*{Acknowledgments}
Xiaoen Zhang is supported by the National Natural Science Foundation of China (Grant No.12101246); Liming Ling is supported by the National Natural Science Foundation of China (Grant No. 12122105) and Guangzhou Science and Technology Plan (No. 2024A04J6245).

%\subsection*{Conflict of interest}
%The authors declare that they have no conflict of interest with other people or organization that may inappropriately influence the authors work.
%\subsection*{Data availability statement}
%Data sharing was not applicable to this article as no datasets were generated or analyzed during the current study.

\bibliographystyle{siam}
\bibliography{reference}
\end{document}